\documentclass[11pt]{article}
\usepackage[margin=1in]{geometry}
\usepackage{setspace}
\usepackage{amsfonts,mathtools,mathtools,color}
\usepackage[dvipsnames]{xcolor}
\usepackage[colorlinks=true,
            citecolor=MidnightBlue,
            linkcolor=RedViolet,
            urlcolor=blue]{hyperref}
\usepackage{algorithm}
\usepackage{algpseudocode}
\usepackage{bm}
\usepackage{amsmath}
\usepackage{amssymb}
\usepackage{amsthm}
\usepackage{wrapfig}
\usepackage{bbm}
\usepackage{bbold}
\usepackage{xcolor}
\usepackage{comment}
\usepackage{graphicx}
\usepackage{subcaption}
\usepackage{tcolorbox}
\usepackage[autostyle]{csquotes}
\numberwithin{figure}{section}
\allowdisplaybreaks
\newcommand{\R}{\mathbbm{R}}
\newcommand{\Z}{\mathbbm{Z}}
\newcommand{\Q}{\mathbbm{Q}}

\newtheorem{theorem}{Theorem}[section]
\newtheorem*{theorem*}{Theorem}
\newtheorem{claim}[theorem]{Claim}
\newtheorem{lemma}[theorem]{Lemma}
\newtheorem{prop}[theorem]{Proposition}
\newtheorem{definition}[theorem]{Definition}
\newtheorem{observation}[theorem]{Observation}

\newcommand{\vol}{\mathsf{vol}}
\newcommand{\eps}{\varepsilon}
\newcommand{\del}{\delta}
\newcommand{\seq}{\subseteq}
\newcommand{\m}{\setminus}
\newcommand{\KNAP}{\text{KNAP}}
\newcommand{\cl}{\text{cl}}
\newcommand{\from}{\leftarrow}
\newcommand{\W}{\mathcal{W}}
\newcommand{\M}{\hat{M}}
\newcommand{\ind}{j}
\newcommand{\cub}{C}
\newcommand{\1}{\mathbbm{1}_n}
\newcommand{\0}{\mathbb{0}_n}
\newcommand{\is}[1]{\mathbf{1}_{#1}}
\newcommand{\el}{l}

\begin{document}

\title{Deterministic Volume Estimation of Truncated Hypercubes
}
\author{Kyra Gunluk}
\date{}

\begin{titlepage}
\maketitle
\thispagestyle{empty}

\begin{abstract}
We present a deterministic polynomial-time algorithm for estimating the volume of a hypercube intersected by a fixed number of constraints, where each constraint is the sum of univariate functions that are each nonnegative, monotone, and convex. Such constraints include knapsack and norm-ball constraints. The case of the unit hypercube truncated by a single linear constraint (halfspace) is already \#P-hard. Given a fixed number of such constraints in a given dimension, together with an error parameter, our algorithm computes a multiplicative approximation of the volume of their intersection with the unit hypercube in time polynomial in the dimension, reciprocal error, input encoding length, and output encoding length.
\end{abstract}

\end{titlepage}

\setcounter{page}{1}

\section{Introduction}
Computing the volume is an ancient and difficult problem, even for convex bodies. Dyer and Frieze ~\cite{DF88} showed that computing exact volume is \#P-hard, even for an explicit polyhedron specified by a unimodular constraint matrix.
In fact, even approximating the volume of a convex body is notoriously difficult. In the general membership oracle model, it was shown by \cite{E86,Barany1987,BF88} that any deterministic algorithm that computes a polynomial (in the dimension) relative approximation of the volume of a convex body must incur exponential complexity. 
Against this backdrop, the randomized polynomial-time approximation scheme of Dyer, Frieze and Kannan~\cite{DyerFK89} was a surprising breakthrough that heralded an age of new techniques for randomized algorithms and analysis for volume computation that have also been extended to logconcave integration~\cite{L90,AK91,DyerF90,LS90,LS93,KLS95,KLS97,LV06,Lovasz2007,KV06,CousinsV18,kookV2025}. 

Our interest here is in the setting where the input is specified explicitly (and therefore the membership oracle model lower bounds for deterministic algorithms do not apply). For these explicitly specified polyhedra, the known efficient algorithms are essentially the same as in the membership oracle model and rely heavily on randomization --- they are based on sampling a sequence of distributions by Markov chains. The main motivation for the present paper is to understand classes of convex bodies whose volume can be estimated by efficient \emph{deterministic} algorithms. 

Going back to the original hardness proof of Dyer and Frieze~\cite{DF88}, the core problem shown to be \#P-hard is very simple: Let $P$ be the explicit polytope obtained by intersecting the unit hypercube $[0,1]^n$ with a {\em single} halfspace described by a nonnegative constraint vector, $a^\top x\le b$. Computing the volume of this family of polytopes, equivalently, computing the probability that a random point in a unit hypercube satisfies a given linear constraint, is already \#P-hard.

\begin{theorem}[Polytope Volume is \#P-hard~\cite{DF88}]\label{Dyer-Frieze}
Let $P(A,b) = \{ x \in \mathbb{R}^n : Ax \le b \}$ be a polyhedron.
$A, b$ have rational entries where $A = (a_{ij})$, $i = 1,2,\dots,k$,
$j = 1,2,\dots,n$, and $b = (b_i)$, $i = 1,2,\dots,k$.
    Computing $\operatorname{vol}(P(A,b))$ is \#P-hard even when $A$ is totally unimodular; it is also \#P-hard for $P = [0,1]^n \cap \{x: a^\top x \le b\}$ for a given integer vector $a$.  
\end{theorem}
Despite the hardness of exact computation, the volume of this family of \enquote{Knapsack Polytopes}, $P = [0,1]^n \cap \{x: a^\top x \le b\}$, admits a deterministic fully polynomial-time approximation scheme. Ando and Kijima \cite{andoKijima2016} gave such an FPTAS based on approximate convolution. In fact, they extended this result to knapsack polytopes with a constant number of constraints: $P = [0,1]^n \cap \{x: A^\top x \le b\}$, for nonnegative $A$. We aimed to find a more general class of explicitly specified convex bodies that admit an efficient deterministic volume approximation.

\subsection{Results}
Our main result  provides a deterministic fully polynomial-time approximation scheme for a hypercube truncated by a constant number of decomposable {\em convex} constraints, as depicted in Figure \ref{fig:multi-conv}. These convex bodies include the Knapsack Polytope, as well as a cube intersected with p-norm balls: $\{x : \|x\|_p \leq b\}$ for any $b \geq 0$ and integer $p\geq 1$.

Our results provide deterministic polynomial-time approximation schemes for volume of truncated hypercubes that generalize several special cases studied in the literature.
Existing deterministic algorithms for approximating the volume of convex bodies are either less general or have weaker bounds.
\cite{costandin2024deterministicalgorithmquasipolynomialcomplexity} gives a quasi-polynomial deterministic approximation algorithm for computing the volume of a hypercube intersected with two sets; both of which are of the form $\{x : \sum_{i = 1}^n f_i(x_i) \leq b\}$ where each $f_i$ is a polynomial function, whereas our requirement is that they are nonnegative, nondecreasing, and convex. 
\cite{Barvinok2024} gives a polynomial-time deterministic algorithm for approximating the volume of the nonnegative orthant intersected with $k$ \emph{hyperplanes} (rather than halfspaces); however, their multiplicative error has an exponential dependence on $n-k$.
\cite{guo2024deterministicapproximationvolumetruncated} gives an FPTAS for computing the volume of the truncated fractional matching polytope for a hypergraph with bounded maximum degree: $P = \{ x \in [0,1]^V : x_u + x_v \le 1 \,\,\, \forall \{u,v\} \in E\}
$. As mentioned, \cite{andoKijima2016} gives a FPTAS for a hypercube intersected with $k$ halfspaces with nonnegative normals.

\begin{center}
\begin{minipage}{0.24\textwidth}
\centering
\includegraphics[width=\linewidth]{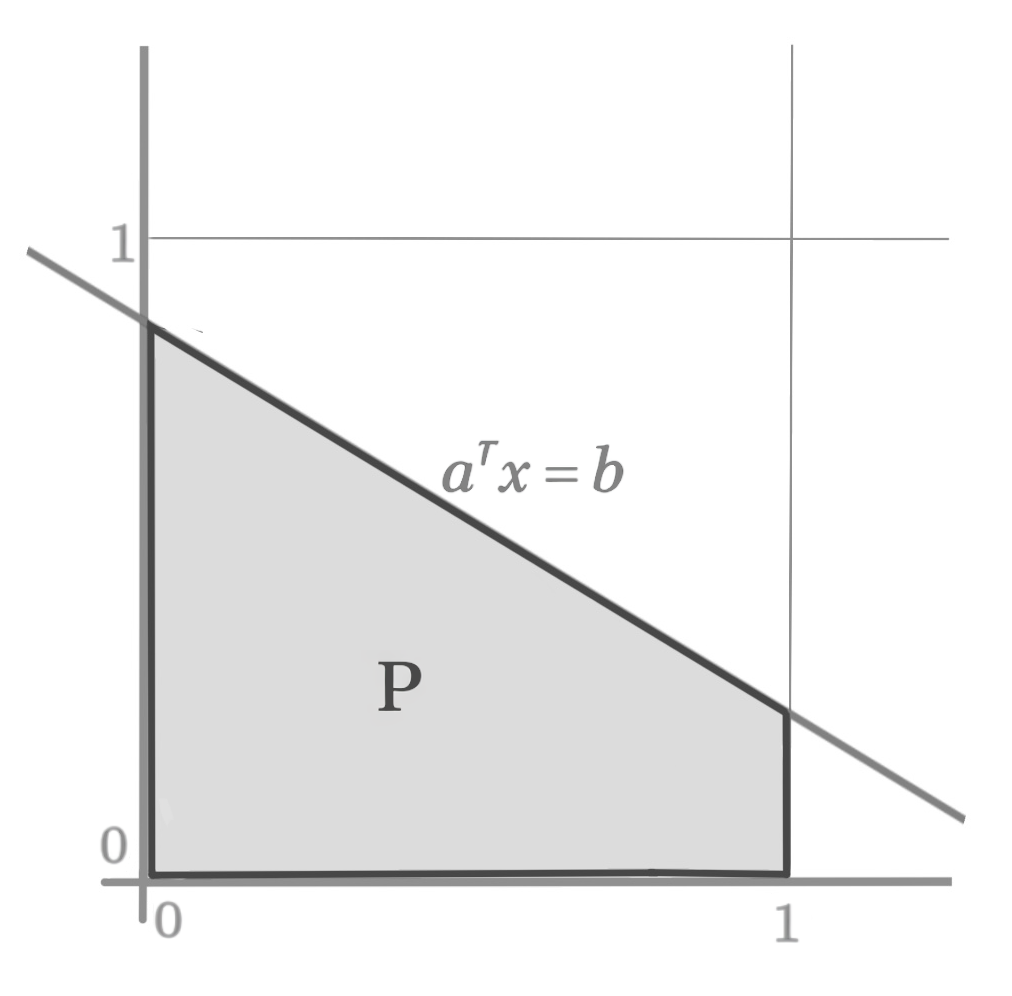}
\vspace{-30pt}
\captionof{figure}{}
\label{fig:one-plane}
\end{minipage}
\hfill
\begin{minipage}{0.24\textwidth}
\centering
\includegraphics[width=\linewidth]{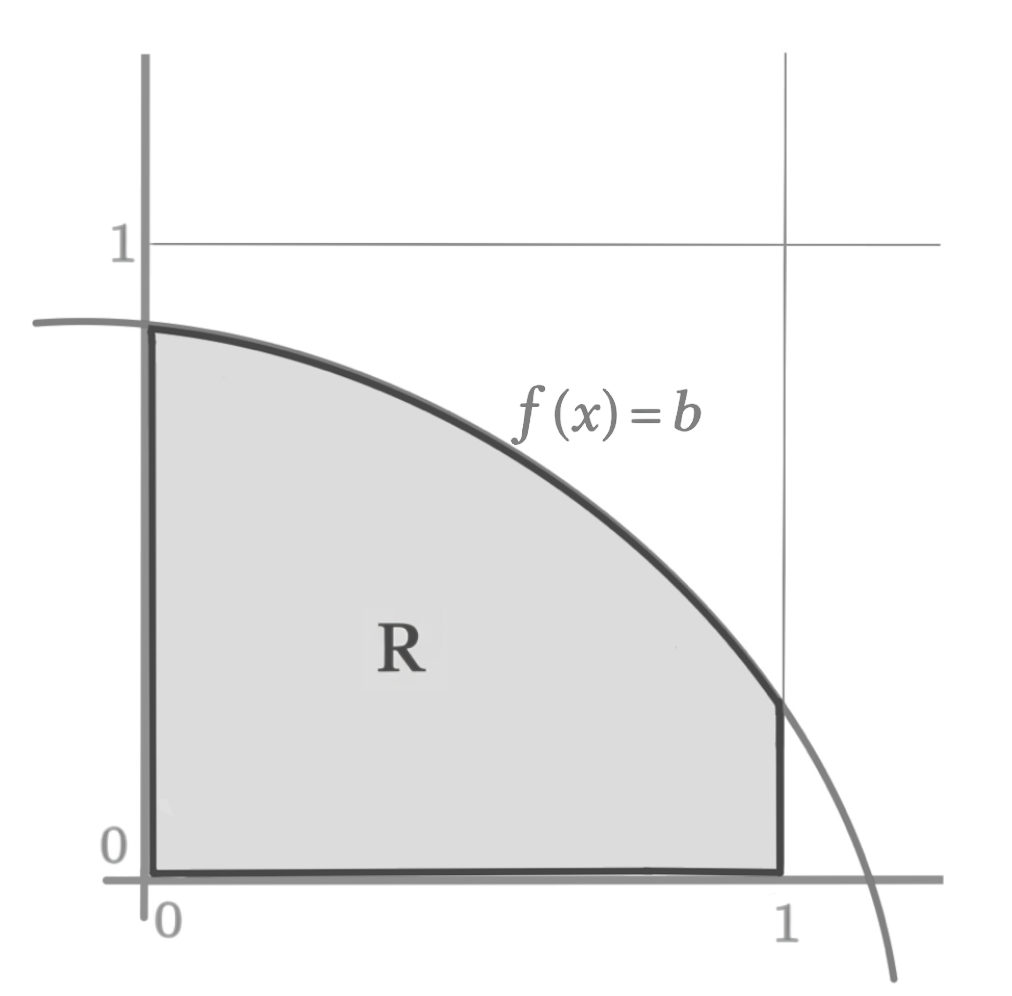}
\vspace{-30pt}
\captionof{figure}{}
\label{fig:one-conv}
\end{minipage}
\hfill
\begin{minipage}{0.24\textwidth}
\centering
\includegraphics[width=\linewidth]{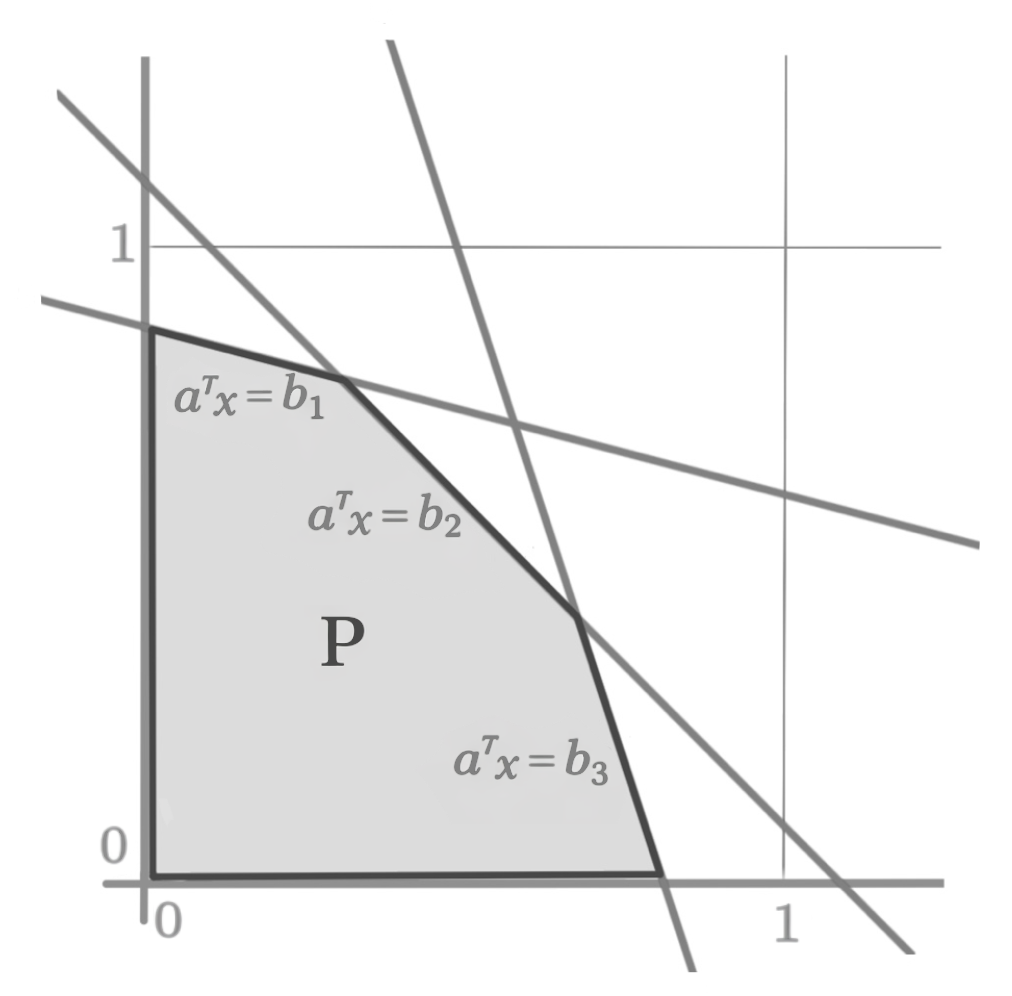}
\vspace{-30pt}
\captionof{figure}{}
\label{fig:multi-plane}
\end{minipage}
\hfill
\begin{minipage}{0.24\textwidth}
\centering
\includegraphics[width=\linewidth]{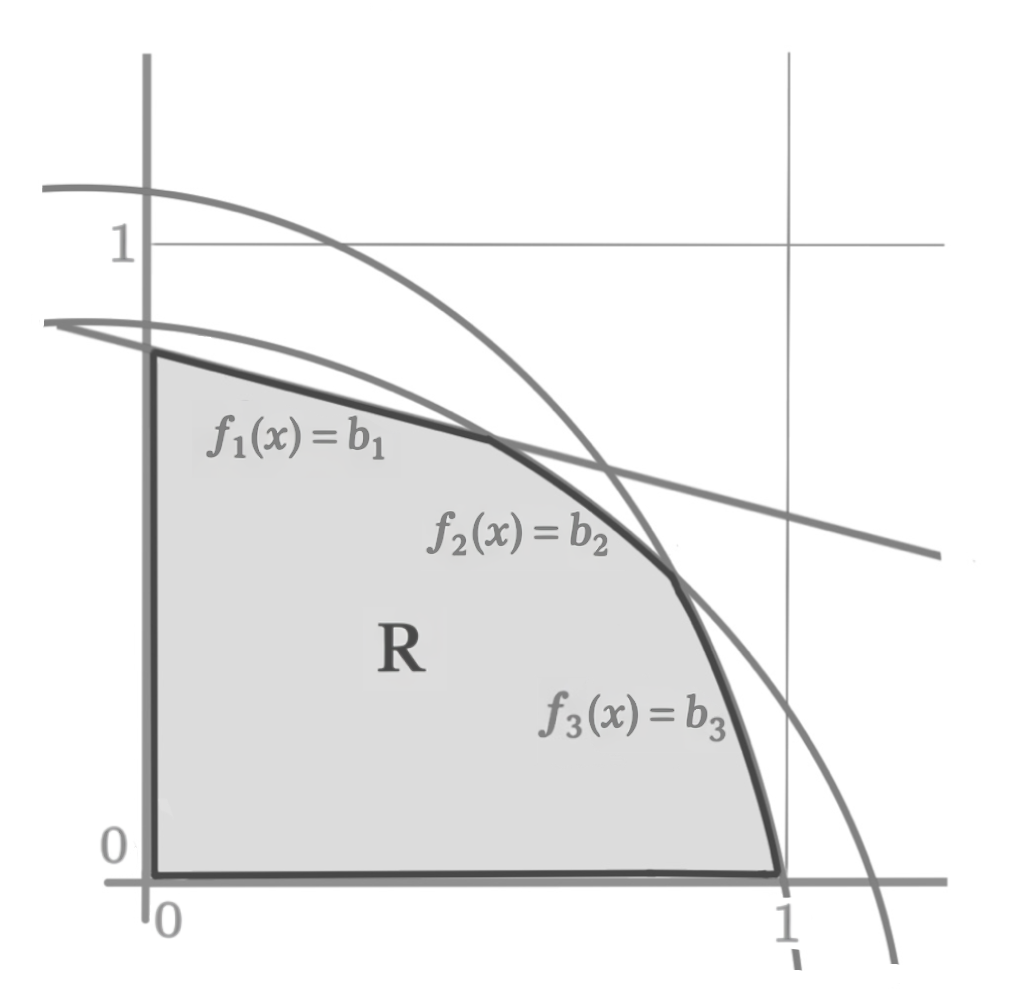}
\vspace{-30pt}
\captionof{figure}{}
\label{fig:multi-conv}
\end{minipage}
\end{center}

We first design an FPTAS for our class of cubes intersected by a single constraint, as depicted in Figure \ref{fig:one-conv}, through the following theorem:
\begin{theorem}\label{thm:one-convex}
    Let $Q = [0,1]^n \cap \{x : \sum_{j = 1}^nf_j(x_j) \leq b\}$ where  $b\in \Q$, $b > \sum_{j = 1}^n f_j(0)$ and $f_j: \R \to \R_{\geq 0}$ are nonnegative, nondecreasing convex functions such that on input $x_j \in \Q$ then $f_j(x_j) \in \Q$ for all $j \in [n]$. Given $\eps >0$, there exists a deterministic algorithm that computes $Z' \in \R$ such that \[\vol(Q) \leq Z' \leq (1 + \eps)\vol(Q)\] 
        using $O(n^5 (L_o+\log({n/}\eps))^3 (L + \log n) / \eps^2)$ 
        arithmetic operations on $O(L(\log(\frac{n}{\eps})+ L_o))$-bit numbers, where $L$ is the max input encoding length and $L_o$ is the max output encoding length, i.e. the encoding length of $\vol(Q)$.
\end{theorem}
Note that a $p$-norm ball with bounded positive integer $p$ is a candidate for function $f$ because it has the form $\|x\|_p \leq b \iff \sum_{i = 1}^n x_i^p \leq b^p$. Thus on input $x \in \Q^n$ with encoding length $L$, each $f_j(x_j) = x_j^p$ remains rational with encoding length $O(pL)$. 

We then extend this FPTAS to include cubes intersected by a constant number of these constraints:
\begin{theorem}\label{thm:multi-conv}
    Let $Q = [0,1]^n \cap \{x : \sum_{j = 1}^nf_{ij}(x_j) \leq b_i \quad \forall i \in [k]\}$ where each $b_i \in \Q$, $b_i > \sum_{j = 1}^n f_{ij}(0)$ and $f_{ij} : \R \to \R_{\geq 0}$ are nonnegative, nondecreasing convex functions such that on input $x_j \in \Q$, then $f_{ij}(x_j) \in \Q$ for all $i \in [k],j \in [n]$. Given $\eps >0$, there exists a deterministic algorithm that computes $Z' \in \R$ such that \[\vol(Q) \leq Z' \leq (1 + \eps)\vol(Q)\] using $n^{O(k^2)}(\log \frac{n}\eps + L_o)^{O(k)}(L + \log n)/\eps^{O(k)}$ arithmetic operations on $O(L(\log (n/\eps) + L_o))$-bit numbers, where $L $ is the max input encoding length and $L_o$ is the max output encoding length.
\end{theorem}

In our work, we prove a more specific case of this truncated hypercube to build towards the algorithm in our general case. We begin with 
the family of polytopes used in the \#P-hardness proof by Dyer and Frieze; a hypercube truncated by a single arbitrary halfspace, as depicted in Figure \ref{fig:one-plane}. Note that this result is redundant of the main theorem of \cite{andoKijima2016}, and we only include it to work towards understanding the algorithm.
\begin{theorem}\label{thm:one-plane}
    Let $P = [0,1]^n \cap \{x : a^\top x \leq b\}$ where $a \in \Z^n, b\in \Z$. Given $\eps >0$, there exists a deterministic algorithm that computes $Z' \in \R$ such that \[\vol(P) \leq Z' \leq (1 + \eps)\vol(P),\] 
    using $O\left(\frac{n^3}\eps (\log\frac{n}\eps + L)^4\right)$ arithmetic operations on $O(L + \log (n/\eps))$-bit numbers, where $L$ is the maximum encoding length of an input parameter.
\end{theorem}
Note that any linear constraint $a^\top x \leq b$ with $a_1, \dots, a_n, b \in \Q$ intersecting a cube of any length can be transformed into a polytope of the type above. Such an affine transformation would scale the volume by the determinant of the transformation matrix. Also note that our runtime has a different and incomparable dependence on the dimension size, error parameter, and encoding length than that of \cite{andoKijima2016}, who give a $O(Ln^4/\eps)$ time FPTAS for the same class of polytopes.

We can also extend the result above for a cube intersected with a constant number of nonnegative linear constraints, as depicted in figure \ref{fig:multi-plane}. This is a more specific case of theorem \ref{thm:multi-conv} that we explored in depth to build towards the main theorem. However, we omit the algorithm because it has strictly worse dependence on input parameters than that of \cite{andoKijima2016}, which give a $O((n^2/\eps)^{k+1}nk\log k)$ time FPTAS.

\subsection{Techniques}
The algorithms in this paper rely on choosing a lattice for which the fraction of lattice points in $[0,1]^n$ that lie within our convex body $K$ approximates the volume $\vol(K)$. We find such a lattice by first scaling our body $K$ by a factor of $u \in \Z_{+}$ (it now lies in the $[0,u]^n$ cube) and using points with integer coordinates. For sufficiently large $u$, we can guarantee that the number of integer points is a good approximation of $u^n\cdot \vol(K)$. 
We then adapt known fully polynomial-time approximation schemes for counting integer solutions under constraints $f(x) \leq b$ where $f$ is convex, nonnegative, monotone, and the sum of univariate functions.

It is important to note that we choose $u$ sufficiently large in a way that guarantees that there is no direction in which the body that is too \enquote{narrow}, and thus has few integer points. To do this, we ensure that the minimum axis-intercept (as defined in \ref{axis-intercept}) of our scaled body is \enquote{big enough}. In order to use this notion of axis-intercept, we need some information about where each constraint intersects each axis. This is where our dependence on the encoding length of input and output parameters is crucial. In addition, given constraint $\sum_{j = 1}^nf_{ij}(x_j) \leq b_i$, we will assume that given input $x_j \in \Q$ with encoding length $L_x$, the function value $f_{ij}(x_j)$ has encoding length at most $LL_x$, where $L$ is the maximum input encoding length. This follows from the assumption that each $f_j$ is a \emph{rational} function, i.e. on rational input, it has rational output. Note that this allows us to efficiently evaluate $f_{ij}(x_j)$ on rational inputs, as we are given $f_{ij}$ explicitly.

Now we introduce the main component of each of the volume-estimation algorithms: counting integer solutions in the scaled body.

Our algorithm for Theorem \ref{thm:one-plane} relies on counting knapsack solutions. We call on the following result in \cite{SVV10}.
\begin{theorem}[Stefankovic-Vempala-Vigoda]\label{SVV}
Let $Z = |\{ x \in \{0,1\}^n : w^\top x \leq b \}|$ where $w \in \Z^n_{\geq 0}, b \in \Z_{\geq 0}$. Given $\eps \in (0,1)$, there exists a deterministic algorithm  COUNT\_KNAPSACK($\eps$) that computes $Z' \in \R_{\geq 0}$ such that \[
(1 - \eps) Z \le Z' \le Z.
\]
The algorithm runs in time $O\!\left(n^3 \log(n/\eps)/\eps\right)$.
\end{theorem}

Note that ~\cite{Li_2014} gives a generalization of this theorem, approximating the probability that the sum of integer-valued random variables is at most some value $C$. Their result can be used for our convex constraints in Theorem \ref{thm:one-convex}. However, we turn to a different approach that naturally extends to multiple constraints, which we require for the remainder of our settings.
We use a slight adaptation of the result in \cite{DBLP:journals/corr/abs-1008-3187}, which we prove in the Electronic Companion. 

\begin{theorem}[Extension to Thm 1.3 in \cite{DBLP:journals/corr/abs-1008-3187}]\label{single-conv-ROBP}
Let $Z(f,b,u) =\{ x \in \{0,1, \dots, u-1\}^n : \sum_{j = 1}^nf_j(x_j) \leq b\}$, where each $f_j: \R \to \R_{\geq 0}$ is nondecreasing and rational. Given $\del >0$, there is a deterministic  algorithm ROUND\_ROBP($\del$) that computes $Z' \in \R$ such that 
\[|Z(f,b,u)| \leq Z' \leq (1 + \del)|Z(f,b,u)|. \] using $O(n^5 (\log u)^3 (L + \log n) / \del^2)$ arithmetic operations on $O(L \log u + \log n)$-bit numbers, where $L$ is the maximum input encoding length.
\end{theorem}

For Theorem \ref{thm:multi-conv}, our FPTAS that counts integer points relies on adaptations of two more theorems, both of which we also prove in the Electronic Companion.
\begin{lemma}[Extension to \cite{10.1145/780542.780643}]\label{Dyer}
Let $U_n = \{0,\dots,u_1\}\times \dots \times \{0,\dots,u_n\}.$
Let $Z = \cap_{i \in [k]} Z_i$ where $Z_i =\{ x \in U_n : \sum_{j = 1}^nf_{ij}(x_j) \leq b_i \}$, $b_i > 0$, and each function $f_{ij} : \R \rightarrow \R_{\geq 0}$ is nondecreasing and convex. There exists a set 
\[
S = \cap_{i \in [k]}S_i \quad \mbox{ where } \quad S_i = \bigg\{ x \in U_n : \sum_{j = 1}^n \bigg\lfloor \frac{2n^2f_{ij}(x_j)}{b_i}\bigg\rfloor \leq 2n^2 \bigg\}
\]
such that $|Z|\leq |S| \leq 2(2n)^k|Z|$.
\end{lemma}

\begin{theorem}[Extension to Theorem 1.2 of \cite{DBLP:journals/corr/abs-1008-3187}]\label{multi-ROBP}
Let $Z = \cap_{i \in [k]} Z(f_i,b_i,u)$ where $Z(f_i,b_i,u) =\{ x \in \{0,1,\dots,u-1\}^n : \sum_{j = 1}^nf_{ij}(x_j) \leq b_i\}$, $b_i > 0$ and each $f_{ij} : \R \rightarrow \R_{\geq 0}$ is nondecreasing, convex, and rational. Given $\del > 0$, there is a deterministic algorithm ROUND\_ROBPS($\del$) that computes $Z'$ such that 
\[|Z|\leq Z' \leq (1 + \delta)|Z|\]
using $O(n^{O(k^2)}(\log u/\del)^{O(k)}(L + \log n))$ arithmetic operations on $O(L\log u + \log n)$-bit numbers, where $L$ is the max input encoding length. 
\end{theorem}

The assumption that the normals of the truncating hyperplanes have nonnegative coefficients is necessary for the techniques we use, i.e., 
one cannot directly extend the method of approximating volume through approximately counting integer points when coefficients are allowed to be arbitrary. This is because there is no $\eps$-multiplicative approximation for the cardinality of such a set, unless $P = NP$. More precisely, we have the following result:

\begin{theorem}\label{thm: hardness}
Given $a\in \Z^n, b \in \Z$, determining whether the set $Z = \{x \in \{0,1\}^n \,:\, -b \le a^\top x \le b\}$ has a non-zero solution is NP-hard and hence given $\eps < 1$ we cannot efficiently compute $Z' \in R$ such that $|Z| \leq Z' \leq (1 + \eps)|Z|$, unless P=NP. 
\end{theorem}

\section{Preliminaries}
\begin{wrapfigure}[6]{r}{0.25\textwidth}
\centering
\vspace{-50pt}
\includegraphics[width=\linewidth]{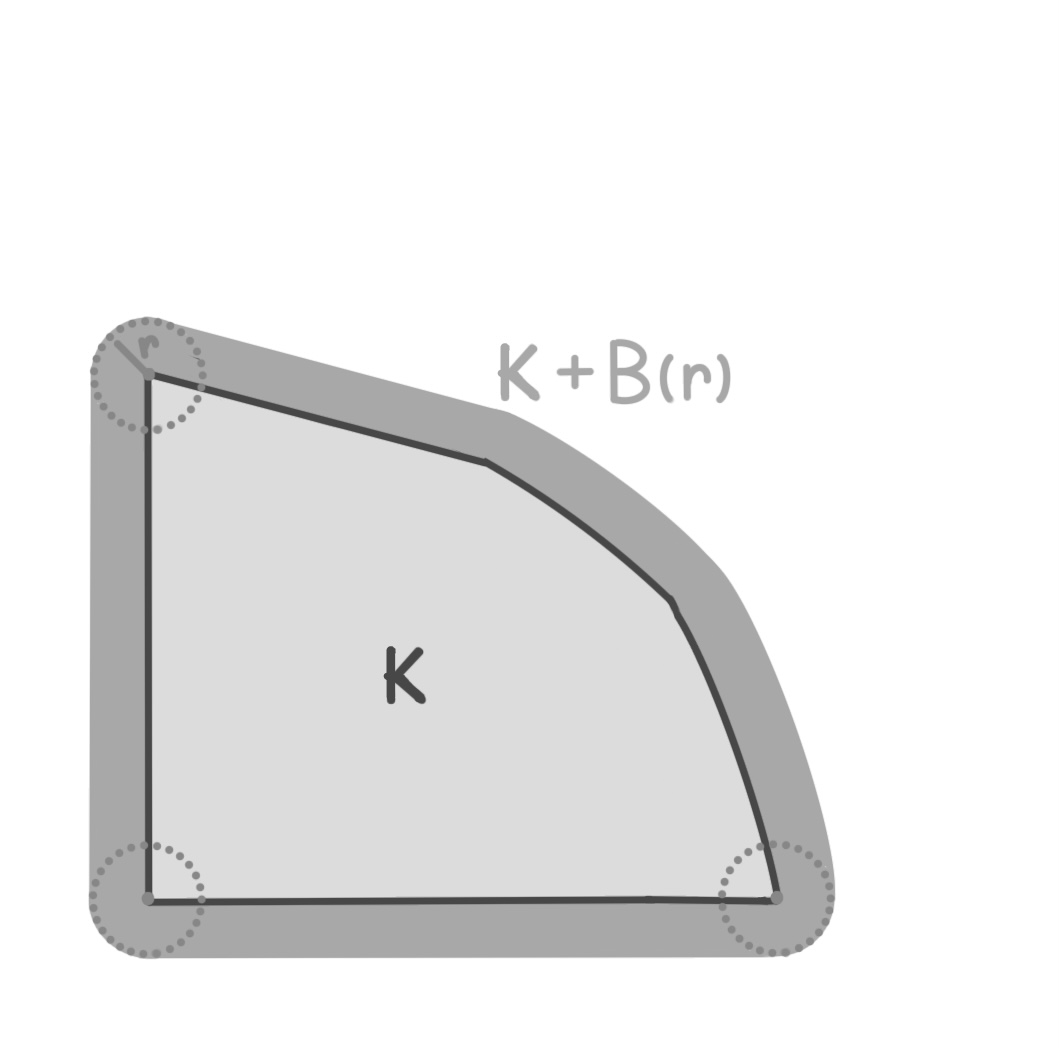}
\vspace{-25pt}
\caption{}
\label{fig:rounded-convex-body}
\end{wrapfigure}
The following notation will be used throughout the paper:
\begin{itemize}
    \item \textbf{Closure:}
    We denote the closure of set $S$ as $\cl(S)$.
    \item \textbf{Rounded Convex Body:}
    For a convex body $K$, $r \in \R_{\geq 0}$, and the Euclidean distance metric $d: \R^n\times\R^n\to \R_{\geq0}$, $\boldsymbol{K+ B(r)}$ is the convex body $\{y: \exists x \in K \text{ with } d(x,y) \leq r\}$. See figure \ref{fig:rounded-convex-body}.
    \item \textbf{Dilated Convex Body:}
    Given a convex body $K$ and $r \in \R_{\geq 0}$, $\boldsymbol{rK}$ is the convex body $\{rx: x \in K\}$. Note that $\vol(rK) = r^n\vol(K)$.
\end{itemize}

\begin{itemize}
    \item For $x, y \in \R^n$, we say $x \leq y$ if $x_i \leq y_i$ for all $i \in [n]$. Similarly, $x < y$ if $x_i < y_i$ for all $i \in [n]$.
    \item We call a function $f:\R^n\to \R$ \textbf{ nondecreasing} if for every $x \leq y \in \R^n$ we have $f (x) \leq f(y)$.
    \item We use $\1$ to refer to the length $n$ vector where each entry is a $1$. Similarly, $\0$ refers to the length $n$ vector where each entry is a $0$.
    \item We use $e_j \in \R^n$ to refer to the length $n$ vector with entries $(e_j)_i = \begin{cases}
        1 & \text{ if }i = j \\ 0 &\text{ otherwise}
    \end{cases}.$
    \item $\is{S} \in \R^n$ is the indicator vector of set $S \seq [n]$. The vector has entries $(\is{S})_j = \begin{cases}
        1 & \text{ if }j \in S \\ 0 &\text{ otherwise}
    \end{cases}.$
    \item Given vector $x \in \R^n$ and scalar $r \in \R$, $rx$ is the vector where every entry is multiplied by $r$, and $x/r$ is the vector where every entry is divided by $r$.
    \item We say a function $f$ is \textbf{rational} if, on rational input, it has rational output.
    \item \label{bitsize: 1} The \textbf{encoding length}, of a rational input $p/q$ is the number of bits necessary to represent $p/q$: $\lceil\log(|p| + 1)\rceil + \lceil\log(q + 1)\rceil$
\end{itemize}
We will also reference the following throughout the paper:
\begin{definition}[Axis-Intercept]\label{axis-intercept}
    Let $K \seq \R_{\geq 0}^n$ be a closed convex body such that $\0 \in K$. For fixed $j \in [n]$, the \textbf{$\bm{x_j}$-axis-intercept} of $K$ is the value $\ell_j(K) \in \R_{\geq 0}$ such that $\ell_j(K)e_j \in K$, but for any $y > \ell_j(K)$, $ye_j \not\in K$.
    The \textbf{minimum axis-intercept} is $\ell(K) := \min_{j \in [n]} \ell_j(K)$.
    Note that for $r \geq 1$, dilation $rK$ has minimum axis intercept $\ell(rK) =  r\ell(K)$.
\end{definition}
\begin{observation}\label{obs}
Consider $K = \{ x \in [0,u]^n : \sum_{j = 1}^nf_{ij}({x_j}/u) \leq b_i \quad \forall i\in [k]\}$ where each $f_{ij}: \R\to \R_{\geq 0}$ is nonnegative, nondecreasing, and convex. 
\begin{itemize}
    \item Note that $h_{ij}(x_j) = f_{ij}({x_j}) - f_{ij}(0)$ is also nonnegative, nondecreasing and convex and
    \[\sum_{j = 1}^nf_{ij}\left(\frac{x_j}u\right) \leq b_i \iff \sum_{j = 1}^nh_{ij}\left(\frac{x_j}u\right) \leq b_i- \sum_{j = 1}^nf_{ij}\left(0\right)\]
    Thus we can assume without loss of generality that $f_{ij}(0) = 0$ for all $i \in [k]$, $j \in [n]$.
    \item If $\0 \not\in K$, or equivalently, $b_i < \sum_{j = 1}^nf_{ij}({0}/u) = 0$ for some $i \in [k]$, then $K = \emptyset$, and $\vol(K) = 0$. Thus throughout the paper we will assume this is not the case.
    \item $K$ is the intersection of level sets of convex functions and a convex body, thus $K$ is convex. 
    \item If $\ell(K) = 0$, then $K$ is not full-dimensional, and $\vol(K) = 0$. We will again assume this is not the case.
\end{itemize}.
\end{observation}
\vspace{-25pt}
Our algorithms will heavily rely on the following theorems.
\begin{theorem}\label{find-l}
    Let $K = \{ x \in [0,1]^n : \sum_{j = 1}^nf_{ij}(x_j) \leq b_i \quad \forall i\in [k]\}$ where each $b_i > \sum_{j = 1}^n f_{ij}(0)$ and each $f_{ij}: \R \to \R_{\geq 0}$ is nonnegative, nondecreasing, convex and rational.
    Let $L_o$ be the encoding length of $\vol(K)$. There exists an algorithm FIND\_INTERCEPT that finds $\ell' \in \R_{\geq 0}$ such that 
    \[\frac 12\ell(K) \leq \ell' \leq \ell(K)\]
    using $O(nL_o)$ arithmetic operations on $O(LL_o)$-bit numbers. 
\end{theorem}

\begin{proof}[Proof of Theorem \ref{find-l}]
    To find an approximate minimum axis-intercept $\ell'$ of $Q$, we will first find an approximate $x_j$-axis-intercept $\ell'_j$ for each $j \in [n]$, and then take the minimum over each approximate intercept.
    
    For fixed $j \in [n]$, we perform a repeated halving on $(0,1]$.
    Starting with $z = 1$, we will evaluate $f_{ij}(z)$, and terminate if this value does not exceed $b_i$ for all $i \in [k]$. Otherwise we take $z \from z/2$ and repeat.
    In other words, at iteration $m$, we compute $f_{ij}\big(1/2^{m-1})$ until this term is $\leq b_i$ for every $i$, i.e. $e_j/2^{m-1} \not\in K$. We take ${\ell_j}' = 1/2^{m-1}$.

    After repeating this process for all $j \in [n]$, we return $\ell' = \min_{j \in [n]} \ell'_j$.
    
    We first confirm that this $\ell'$ satisfies the conditions desired. Note that either $\ell'_j = \ell_j(K) = 1$, or $e_j \not\in K$. In the latter case, the last iteration $m$ satisfies $e_j/2^{m-1} \in K$, but $e_j/2^{m-2} \not\in K$. It must follow in either case that $1/2^{m-1} \leq \ell_j(K) \leq 1/2^{m-2}$, and thus \[\frac12{\ell_j(K) \leq \frac12 \cdot \frac1{2^{m-2}}} = \ell'_j = \frac1{2^{m-1}} \leq \ell_j(K).\]
    Consequently, by taking $\ell' = \min_{j \in [n]} \ell'_j$, we have $\ell' \leq \ell'_j \leq \ell_j(K)$ for every $j \in [n]$, and thus $\ell' \leq \min_{j \in [n]}\ell_j(K) = \ell(K)$. Similarly,
    $\ell(K)/2 \leq \ell_j(K)/2 \leq \ell'_j$ for every $j \in [n]$, and thus $\ell(K)/2 \leq \min_{j \in [n]}\ell'_j = \ell'$. Thus we have the proposed condition
    \[\frac12 \ell(K) \leq \ell' \leq \ell(K).\]
    
    Now we analyze the runtime of the algorithm. 
    Note that we repeat our process of picking $\ell'_j$ $O(n)$ times, once for each $j$. 

    Consider the true minimum axis intercept $\ell(K) = \ell_{j^*}(K)$ for some ${j^*} \in [n]$. Note that for any $y > \ell(K)$, we have $ye_{j^*} \not\in K$, and thus by monotonicity of each function, it follows that for any $x \geq ye_{j^*}$, $x \not \in K$. Consequently, any $x \in K$ must have $x_{j^*} \leq \ell(K)$, thus $K$ is contained in the hyperrectangle with length $\ell(K)$ in dimension $x_{j^*}$, and length $1$ in every other dimension. Note that this hyperrectangle has volume $\ell(K)$, and thus $\vol(K) \leq \ell(K)$. It follows that 
    \[ \frac12\vol(K) \leq \frac12 \ell(K) \leq \frac12\ell_j(K) \leq \ell'_j \qquad \forall j \in [n]\]
    Also note that the encoding size, $L_o$, of $\vol(K)$ is at least $\log(1 / \vol(K))$, since $\vol(K) \leq 1$.
    Thus, for each $j \in [n]$, after $m = L_o + 1 = \lceil\log(1 / \vol(K)) + 1 \rceil + 1 \geq \log(1/\vol(K)) + 1$ iterations, we have $1/2^{m-1} \leq \vol(K) \leq \ell'_j$, and thus our search needs at most $O(L_o)$ iterations. In each iteration, and for each of $i \in [k]$, we compute function value $f_{ij}(x_j)$ where $x_j$ uses at most $L_o$-bits and thus $f_{ij}(x_j)$ uses at most $LL_o$-bits. 
    
    Thus we perform $O(nL_o)$ arithmetic operations on $O(LL_o)$-bit numbers.
\end{proof}

\begin{prop}\label{Cube containment}
Let $b \in R^k_{\geq 0}$, $u \in \Z_+$, and $f_{i}: \R^n\to \R_{\geq 0}$ be nonnegative, nondecreasing, and convex functions for $i\in [k]$. 
Let 
$K = \{x\in[0,u)^n : f_{i}(x/u) \leq b_i \quad \forall i\in [k]\} $.
Consider the axis-aligned partition of $[0,u)^n$ into unit cubes, and let $C$ be the subset of these cubes that intersect $K$ (See figure \ref{fig:cubes}).
Let $\ell = \ell(\cl(K))$ be the minimum axis-intercept of $\cl(K)$. 
Then 
\[
\vol(K) \leq \vol(C) \leq \left(1 + \frac{2n\sqrt{n}}{\ell}\right)^n \vol(K).
\]
\end{prop}

\begin{proof}[Proof of Proposition \ref{Cube containment}]
Let $Z = K \cap \Z^n$ be the integer points in $K$, and let $C' = \cup_{z \in Z} C_z$ where $C_z = \{x \in \R^n : z \leq x < z + \1\}$. We will first show that $C' = C$. Clearly, $C'$ does not include any axis aligned cubes that do not intersect $K$, as each cube has \enquote{root} $z \in Z \seq K$. Now we show there is no point in $K$ that is not in some $C_z \seq C'$.
Consider any $x \in K$, and let $\bar{x}$ be such that $\bar{x}_i = \lfloor x_i\rfloor$ for all $i \in [n]$. Note that $\bar{x} \leq x < \bar{x} + \1$. Since $x \in [0,u)^n$, it follows that $\bar{x} \in [0,u-1]^n\seq[0,u]^n$. By the nondecreasing property of each $f_{i}$, we have $f_{i}({\bar{x}}/u) \leq f_{i}({x}/u) \leq b_i$ for all $i$. Finally, $\bar{x} \in \Z^n$ by construction, and consequently $\bar{x} \in Z$. Thus $x \in \cub_{\bar{x}} \seq \cub'$, and $\cub = \cub'$.

Let $r = {\ell}/{2n}$. We will prove that \[K \seq C \seq (K + B(\sqrt{n}))\cap \R^n_{\geq 0} \seq \left(1 + \frac{\sqrt{n}}{r}\right)\cl(K),\]
See Figure \ref{fig:cubes} for an example.

It is clear that $K \seq C$, since we include all cubes that intersect $K$. 

Next we see that $C \seq (K + B(\sqrt{n}))\cap \R^n_{\geq 0}$:

    Consider any $x \in C$, and take $\bar{x} \in Z \seq K$ such that $x \in C_{\bar{x}}$. Note that since $C_{\bar{x}}$ has length $1$, the furthest distance within this cube is $\sqrt{n}$, and thus the distance between $x$ and $\bar{x}$ (a point in $K$) is at most $\sqrt{n}$. Thus $x \in K + B(\sqrt{n})$. Also, note that since $\bar{x} \in K \seq [0,u)^n$, we have $\0 \leq \bar{x}$, and thus $\0 \leq \bar{x} \leq x$, and so $x \in \R^n_{\geq 0}$. Thus $x \in (K + B(\sqrt{n}))\cap \R^n_{\geq 0}$.

Now we argue that $(K + B(\sqrt{n}))\cap \R^n_{\geq 0} \seq (1 + {\sqrt{n}}/{r})\cl(K)$.

It is clear that $K \seq (1 + {\sqrt{n}}/{r})K\seq (1 + {\sqrt{n}}/{r})\cl(K)$, so it remains to show that every $x \in (K + B({\sqrt{n}}))\cap\R_{\geq 0}^n \m K=((K + B({\sqrt{n}})) \m K)\cap\R_{\geq 0}^n$ is also in $(1 + {\sqrt{n}}/{r})\cl(K)$.

Recall that $\ell = \ell(\cl(K)) = \min_{j \in [n]} \ell_j(\cl(K))$ is the minimum axis-intercept. Observation \ref{obs} tells us that $K$ is convex, thus, for every $j \in [n]$, the entire $x_j$ axis between $0$ and $\ell$ lies within $\cl(K)$, and moreover, $K$ contains the simplex $S := \{x : \sum_{j = 1}^n x_j \leq \ell\}$.
 Consider the ball $B(z,r)$ where $r = {\ell}/{2n}$ and $z_j = {\ell}/{2n}$ for all $j \in [n]$. See figure \ref{fig:containment} for these objects. Note that for any $y \in B(z,r)$, $y_j \leq \frac\ell{n}$ for all $j \in [n]$, and consequently 
    $\sum_{j= 1}^n y_j \leq \sum_{j= 1}^n {\ell}/n = \ell $.
    Thus $y \in S \seq \cl(K)$, and $B(z,r) \seq \cl(K)$.

    Now consider $x \in ((K + B({\sqrt{n}})) \m K)\cap\R_{\geq 0}^n$, and let $x'$ be the point in $\cl(K)$ that has smallest euclidean distance to $x$. Note that $x'$ will be a boundary point of $\cl(K)$. By definition of $K + B({\sqrt{n}})$, $d(x,x') \leq \sqrt{n}$, where $d:\R^n \times \R^n \to \R_{\geq 0}$ is Euclidean distance. By triangle inequality \[d(x,z) \leq d(x',z) + d(x,x') \text{ and thus }\frac{d(x,z)}{d(x',z)} \leq \frac{d(x',z) + \sqrt{n}}{d(x',z)} = 1 +  \frac{\sqrt{n}}{d(x',z)} \leq 1 + \frac{\sqrt{n}}{r}.\]
    where the last inequality comes from the fact that $B(z,r)$ lies fully within $K$, and thus boundary point $x' \in K$ is at least distance $r$ from $z$.
    See figure \ref{fig:distance} for an example.
    
    This allows us to conclude that dilating $K$ by a factor of $1 + {\sqrt{n}}/{r}$ will cover all points in $(K + B(\sqrt{n})) \cap \R^n_{\geq 0}$.
    Thus, $K \seq C \seq (1 + {\sqrt{n}}/{r})\cl(K)$ and, consequently, 
    \[
    \vol(K) \leq \vol(C) \leq \left(1 + \frac{\sqrt{n}}{r}\right)^n\vol(\cl(K)) = \left(1 + \frac{2n\sqrt{n}}{\ell}\right)^n\vol(K).
    \]
\end{proof} 
\noindent
\begin{minipage}{0.3\linewidth}
\centering
\includegraphics[width=\linewidth]{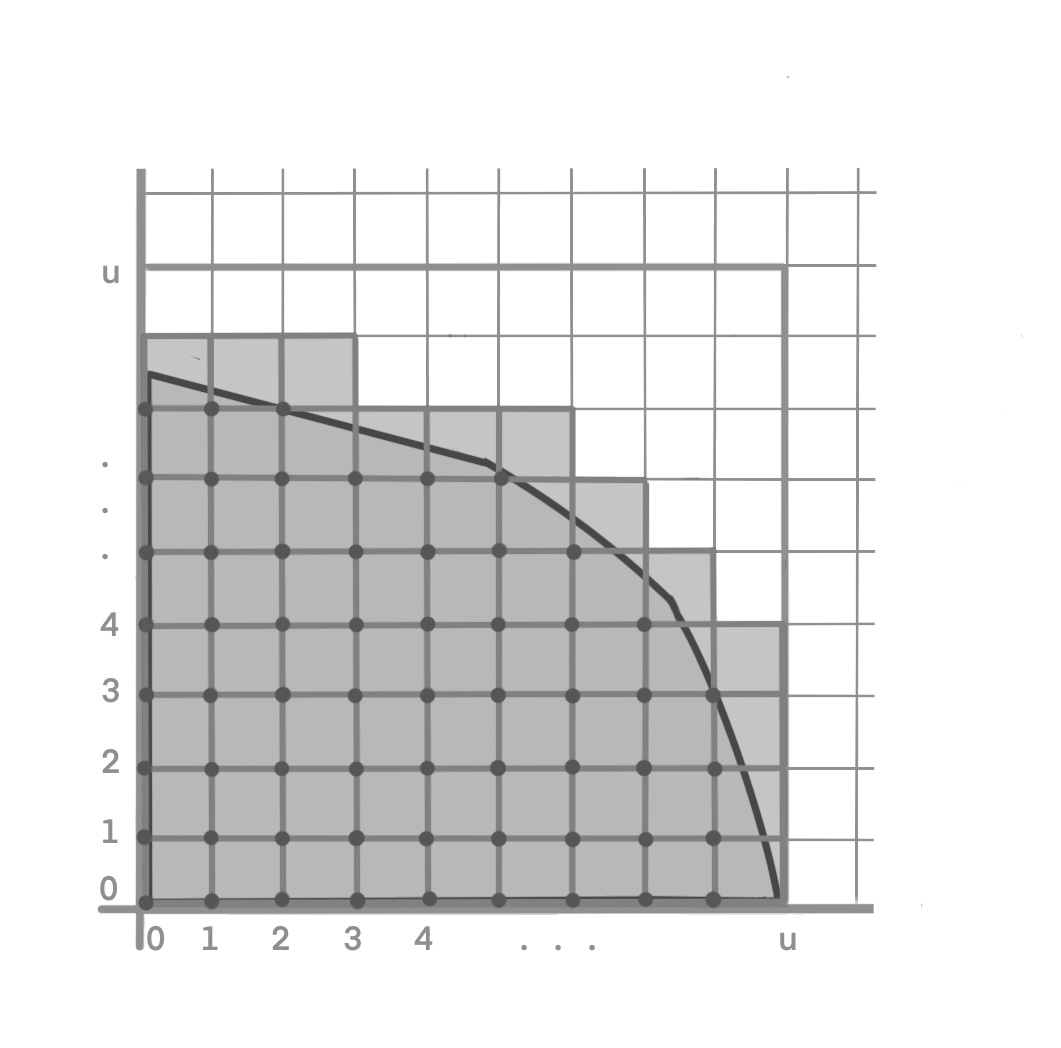}
\vspace{-20pt}
\captionof{figure}{}
\label{fig:cubes}
\end{minipage}\hfill
\begin{minipage}{0.3\linewidth}
\centering
\includegraphics[width=\linewidth]{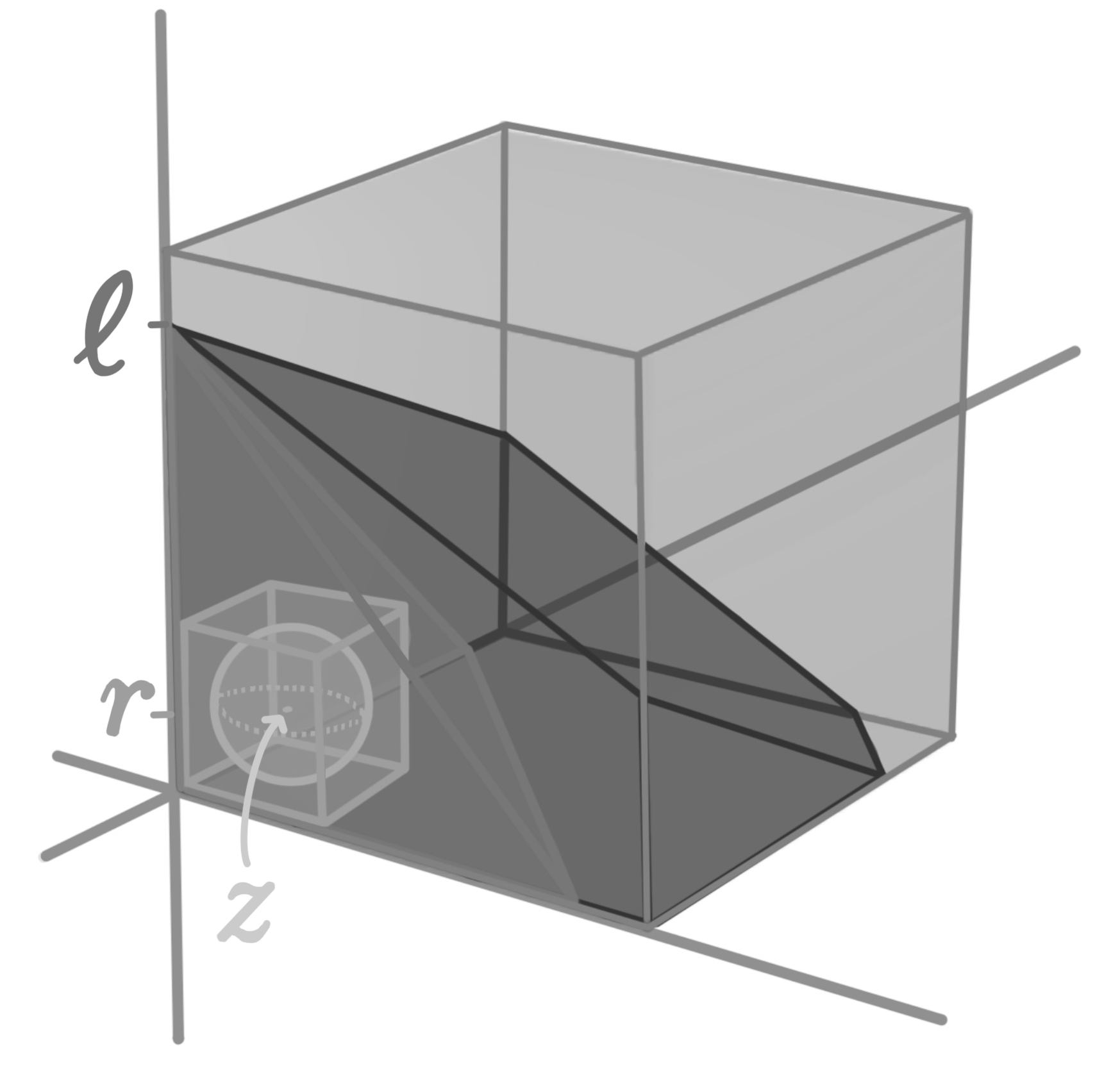}
\vspace{-20pt}
\captionof{figure}{}
\label{fig:containment}
\end{minipage}\hfill
\begin{minipage}{0.3\linewidth}
\centering
\includegraphics[width=\linewidth]{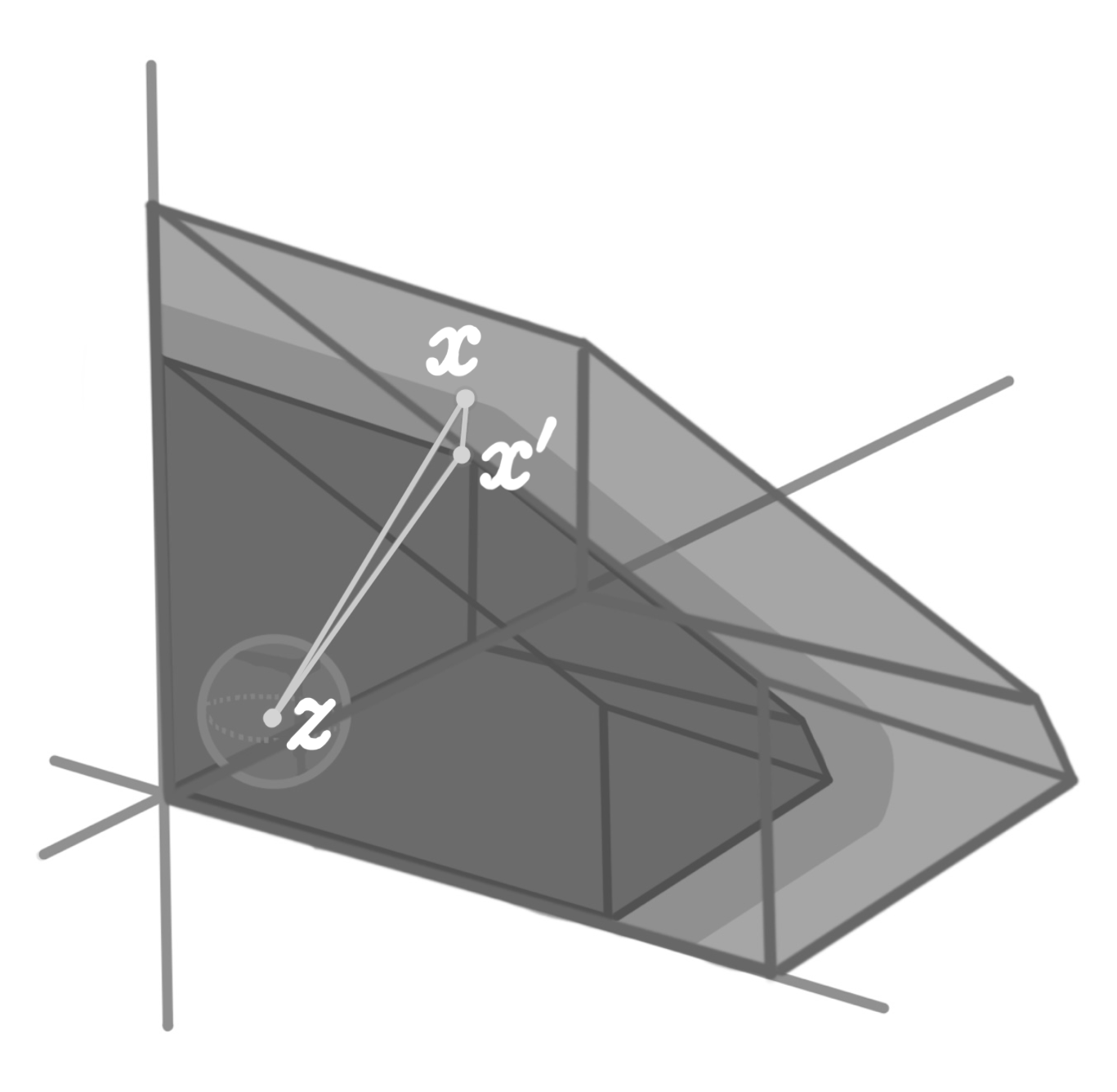}
\vspace{-20pt}
\captionof{figure}{}
\label{fig:distance}
\end{minipage}
\begin{lemma}\label{int-pts-vol}
Let $b \in \R^k$ and let $f_{i}: \R^n \to \R_{\geq 0}$ be nondecreasing, nonnegative and convex functions for $i\in [k]$.
Let 
\begin{align*}
K &=  \bigg\{x \in [0,1)^n : f_{i}({x}) \leq b_i \quad \forall i\in [k]\bigg\} \\
\ell &= \ell(\cl(K)), \,\qquad
u \geq \frac{9n^{2.5}}{\eps\ell}\\
K_u &=  \bigg\{x \in [0,u)^n : f_{i}(\frac1u{x}) \leq b_i \quad \forall i\in [k]\bigg\} \\
Z &= \Z^n \cap K_u.
\end{align*}  
Given any $\eps > 0$ and $Z'\in \R$ such that $|Z| \leq Z' \leq (1 + \frac\eps9)|Z|$, we have $\vol(K_u) \leq Z' \leq (1 + \eps)\vol(K_u)$.
\end{lemma}
\begin{proof}[Proof of Lemma \ref{int-pts-vol}]
    Consider points $Z$ which are all the integer points in $K_u$.
    For any $z \in Z$, we will consider the length $1$ cube rooted at this point: $\cub_z = \{x : z\leq x < z+\1\}$.

    Note that $\cub = \cup_{z \in Z}\cub_z$ is a subset of the partition of $[0,u]^n$ into length $1$ cubes with integer vertices. Thus, 
    \[
    \vol(\cub) = \vol(\cup_{z \in Z} \cub_{z}) = \sum_{z \in Z}\vol(C_z)= \sum_{z \in Z}(1^n) = |Z|.
    \]
Note that $u\ell = u\ell(\cl(K)) = \ell(\cl(K_u))$ is the minimum axis-intercept of $\cl(K_u)$.
    We can apply Prop. \ref{Cube containment} to see
\[\vol(K_u) \leq \vol(C) \leq \left(1 + \frac{2n\sqrt{n}}{u\ell}\right)^n\vol(K_u)\]
From this, we see that for $Z'$ given in the lemma statement, 
\[\vol(K_u) \leq \vol(C) =|Z| \leq Z' \leq (1 + \frac\eps9)|Z| =(1 + \frac\eps9)\vol(C) \leq (1 + \frac\eps9)\left(1 + \frac{2n\sqrt{n}}{u\ell}\right)^n\vol(K_u)\]

Note that $u\ell \geq \frac{9n^{2.5}}{\eps \ell}\ell = \frac{9n^{2.5}}{\eps}$, and thus
\begin{align*}
    Z' &\leq  (1 + \frac\eps9)\left(1 + \frac{2n \sqrt{n}}{u\ell}\right)^n\text{vol}(K_u) \\ 
    &\leq  (1 + \frac\eps9)\left(1 + 2n\frac{2n \sqrt{n}}{u\ell}\right)\text{vol}(K_u) &&\text{By Taylor's series since} \frac{2n^{1.5}}{u\ell} < 1\\
    &\leq  \left(1 + \frac\eps9 + \frac{4n^{2.5} }{\frac{9n^{2.5}}{\eps}} + \frac{4\eps n^{2.5}}{9 \frac{9n^{2.5}}{\eps}} \right)\text{vol}(K_u) \\
    &\leq (1  + \eps )\text{vol}(K_u)
\end{align*} 
Thus $ \vol(K_u) \leq Z' \leq (1 + \eps) \vol(K_u) $.
\end{proof}

\section{Hypercube Clipped by a Single Constraint}
In this section we will prove Theorems \ref{thm:one-plane} and \ref{thm:one-convex}. 

In both the case of linear and convex constraints, our algorithm acts as follows:
We begin by performing affine transformations on the given convex body $K$ to achieve $K_u$ which is the $[0,u]^n$ cube (for some integer $u$) intersected by a constraint $\sum_{i = 1}^ng_j(x_j) \leq b$ where each $g_j : \R \to \R_{\geq 0}$ is nondecreasing and convex on $[0,1]$. We then approximately count the number of integer points in this new body using either COUNT\_KNAPSACK or ROUND\_ROBP. We prove that this provides an approximate number, and thus volume, of axis aligned unit cubes that intersect $K_u$. We then show that this approximation also serves as a good approximation of $\vol(K_u) = u^n \vol(K)$.  

\subsection{Truncating Halfspace}\label{sec: single-knap}
Consider a hypercube $[0,1]^n$ and an intersecting hyperplane $a^\top x = b$. Let polytope $P$ be the portion of the cube clipped by this hyperplane: $P=\{x\in [0,1]^n : a^\top x \leq b\}$.

In order to prove Theorem \ref{thm:one-plane} we will rely on the following propositions:

\begin{prop}
\label{single-knapsack}[Canonical Position]
    There exists an affine transformation that transforms body $P = [0,1]^n \cap \{x : \sum_{j = 1}^n a_j x_j \leq b\}$ to body $R= [0,1]^n \cap \{x : \sum_{j = 1}^n|a_j| x_j \leq  b + \sum_{\ind \in [n] : a_\ind < 0}|a_\ind|\}$ and maintains volume, i.e. $\vol(P) = \vol(R)$.
\end{prop}
\begin{prop}\label{single-int-knap}
     Consider sets $S = \{x \in \R^{n\times \log u} : \sum_{i = 1}^n \sum_{j = 1}^{\log u} 2^{j-1} a_i x_{ij} \leq ub , \,\,\, x_{ij} \in \{0,1\}\} $ and $Z= \Z^n \cap [0,u)^n \cap \{x: a^\top x \leq ub\} =  \left\{x \in \R^n: a^\top x \leq ub, \,\,\, x_i \in \{0,1,\dots,u-1\} \right\}$ where $a,b,u \geq 0$ and $u$ is a power of $2$. Then $|S| = |Z|$.
\end{prop} 
Using the above statements, we can prove the main theorem of this section.

\begin{proof}[Proof of Theorem \ref{thm:one-plane}]
Our algorithm first performs affine transformations on polytope $P$ to achieve polytope $R= [0,1]^n \cap \{x : w^\top x \leq c\}$ where $w_\ind = |a_\ind| \geq 0$ for all $\ind$ and $c = b + \sum_{\ind \in [n] : a_\ind < 0} w_\ind$. Proposition \ref{single-knapsack} tells us that $\vol(P) = \vol(R)$. Thus we can approximate the volume of polytope $R$, which has nonnegative coefficients, instead. From here on we can assume WLOG that $P = [0,1]^n \cap \{x : a^\top x \leq b\}$ with $a_\ind,b \geq 0$ for all $\ind \in [n]$. Note that regardless of the change in the coefficients, the bit complexity of the whole set remains similar: given input encoding length $L$, the new capacity now has encoding length at most $\log(n)+ L$ as it has value at most $b + \sum_{i = 1}^na_j \leq (n+1)(\max\{\max_j a_j, b\})$. 

Let $u$ be the power of $2$ such that \[\frac{9n^{2.5}}{\eps} \max\bigg\{1,{\displaystyle \max_{j \in [n]}\frac{a_j}{b}}\bigg\} \leq u < 2\cdot\frac{9n^{2.5}}{\eps} \max\bigg\{1,{\displaystyle \max_{j \in [n]}\frac{a_j}{b}}\bigg\}.\]

Consider $P_u$ which is the polytope $P$ scaled up by a factor of $u$ in each dimension, such that it now lies in $[0,u]^n$ with hyperplane $a^\top x \leq ub$. We can find $P_u$ by multiplying $P$ with transformation matrix $u I_n$. The volume of the new polytope is exactly \[ \text{vol}(P_u) = det(uI_n) \cdot \text{vol}(P) = u^n\cdot \text{vol}(P) \]
Thus it suffices to develop an FPTAS for the volume of $P_u$. 

Consider the integer points in $P_u':= \{x \in [0,u)^n : a^\top x \leq ub \}$: \[Z = \Z^n \cap [0,u)^n \cap \{x: a^\top x \leq ub\} = \Z^n \cap P_u'\]
Note that $\min\{1,b/a_j\}$ is the $x_j$-axis-intercept of $P$, thus \[u \geq \frac{9n^{2.5}}{\eps} \max\left\{1,{\displaystyle \max_{j \in [n]}\frac{a_j}{b}}\right\} = \frac{9n^{2.5}}{\eps} \frac1{\min\Big\{1,{\displaystyle \min_{j \in [n]}\frac{b}{a_j}}\Big\}} =  \frac{9n^{2.5}}{\eps\ell (P)} .\] Thus we can apply Lemma \ref{int-pts-vol} which tells us that given an $\frac{\eps}9$-approximation $Z'$ for $|Z|$, then $Z'$ serves as an $\eps$-approximation for $\vol(P_u') =\vol(\cl(P_u')) = \vol(P_u)$.

It remains to find an FPTAS that approximately counts the number of integer points in $P_u'$. This is exactly the number of integer knapsack solutions under the knapsack constraint $a^\top x \leq ub$ with $x \in \{0,1,\dots, u-1\}$.

Since $u$ is a power of $2$, Prop \ref{single-int-knap} gives us a set \[S = \bigg\{x \in \R^{n\times\log u} : \sum_{i = 1}^n \sum_{j = 1}^{\log u} 2^{j-1} a_i x_{ij} \leq ub ,\quad x_{ij} \in \{0,1\}\bigg\} \] such that $|S| = |Z|$. Taking matrix $x \in \R^{n \times \log u}$ as a length $n\log u$ vector instead, $S$ can be viewed as a 0-1 knapsack set. 

Note that counting knapsack solutions is a well known $\#P$ hard problem, thus an approximate counting is the best we can do in polynomial time (unless $P = \#P$). 
We rely on the Dynamic Programming FPTAS developed in \cite{SVV10} for approximately counting $0-1$ knapsack solutions in $S$.

Theorem \ref{SVV} gives us $\hat{Z}$ such that $(1-\delta)|S| \leq \hat{Z} \leq |S|$. Taking $\delta = \eps/18$, and $Z' = \hat{Z}/(1-\delta)$, and dividing through by $1-\del$, we find 
\[|S| \leq Z' \leq \frac{1}{1-\del}|S| = \left(\frac{1-\del}{1-\del} + \frac\del{1-\del}\right) |S|= \left(1 + \frac{\eps/18}{1-\eps/18}\right) |S| \leq \left(1 + \frac{\eps/18}{1/2}\right) |S| =  \left(1 + \frac{\eps}{9}\right) |S|\,,\]
thus giving a $\frac{\eps}9$-approximation $Z'$ for $|S| = |Z|$. This FPTAS has runtime $O(\frac{m^3}\del \log (\frac{m}\del))$, where our input parameter $m = n\log u$ and our error parameter $\delta = \frac\eps{18}$. 
Recall that $L \geq \log\{\max_{j \in [n]}a_j\}$ and \[u \leq 2\cdot \frac{9n^{2.5}}\eps\max\bigg\{1,{\displaystyle \max_{j \in [n]}\frac{a_j}{b}}\bigg\} \leq 2 \cdot \frac{9n^{2.5}}\eps\max_{j \in [n]}{a_j},\] and consequently
$\log (u) = O(\log\frac{n}\eps + L)$. Thus our FPTAS uses
$O\left(\frac{n^3}\eps (\log\frac{n}\eps + L)^3 (\log (\frac {n}\eps)  + \log L )\right) = O\left(\frac{n^3}\eps (\log\frac{n}\eps + L)^4\right)$ arithmetic operations. We need not compute values greater than $\max\{u(b + \sum_{j \in [n]}|a_j|), \max_{j \in [n]}a_ju\}$, and thus all operations are on $O(L +\log(n/\eps))$-bit numbers.

This algorithm gives us $Z'$ such that
\[|Z| = |S| \leq Z' \leq (1 + \frac\eps9)|S| = (1 + \frac\eps9)|Z| \] 
and thus
\[\vol(P_u) \leq Z'  \leq (1 + \eps) \vol(P_u)\] 
and 
\[\vol(P) \leq \frac{Z'}{u^n} \leq (1 + \eps) \vol(P)\]
Thus we have an $\eps$ approximation of the volume of our original polytope.

We explicitly state the algorithm below:
\begin{tcolorbox}[colback=gray!10, colframe=gray!60]
\textbf{FPTAS for a Single Truncating Halfspace}

\textbf{ Input:} $a \in \Z^n$, $b \in \Z$, $\eps > 0$
\begin{enumerate}
    \item Set $c \from b + \sum_{j \in [n] : a_j < 0}|a_j|, \quad u \from 2^{\left\lceil \log_2\left(\frac{9n^{2.5}}{\eps} \max\{1,{ \max_{j \in [n]}\frac{|a_j|}{c}}\}\right)\right\rceil}$
    \item $S \from \{x \in \R^{n\log u} : \sum_{i = 1}^n \sum_{j = 1}^{\log u} 2^{j-1} |a_i| x_{i,j} \leq uc , x_{ij} \in \{0,1\}\}$
    \item Call COUNT\_KNAPSACK($\eps/18$) on set $S$, obtain $\hat{Z}$ such that $(1-\eps/18)|S|\leq \hat{Z} \leq |S|$.
    \item Return $\frac{\hat{Z}}{(1 - \eps/18)u^n}$
\end{enumerate}
\end{tcolorbox}
\end{proof}
Now we prove the propositions stated at the beginning of the section
\begin{proof} [Proof of Proposition \ref{single-knapsack}]
Recall that $P = [0,1]^n \cap \{x : a^\top x \leq b\}$ and $R= [0,1]^n \cap \{x : |a|^\top x \leq b + \sum_{\ind \in [n] : a_\ind < 0} |a_\ind|\}$.

Let $J = \{\ind \in [n]: a_\ind < 0\}$. Let $A$ be the diagonal matrix with $A_{\ind,\ind} = -1$ $\forall \ind \in J$ and $A_{\ind,\ind} = 1$ $\forall \ind \not\in J$. Note that $\det(A) \in \{1,-1\}$. Let $\is{J}$ be the indicator vector of the subset $J$. 
Let polytope $Q = \{Ax + \is{J} : x \in P\}$ and notice that $\vol(Q) = |\det(A)| \vol(P) = \vol(P)$.

It remains to show that $Q = R$, or $x \in P \iff Ax + \is{J} \in Q \iff Ax + \is{J} \in R$. 
We prove $x \in [0,1]^n \iff Ax + \is{J}\in [0,1]^n$, and then we prove $x \in \{x : a^\top x \leq b\} \iff Ax +\is{J} \in \{x : |a|^\top x \leq b + \sum_{\ind \in [n] : a_\ind < 0} |a_\ind|\}$, and consequently the claim follows.

Note that 
\[
(Ax+\is{J})_j = \begin{cases}
    -x_j + 1 \quad j \in J \\
    x_j \quad\quad j \not\in J
\end{cases}
\]
and thus $(Ax+\is{J})_j \in [0,1] \iff x \in [0,1]$. 

Next, we see that 
\begin{align*}
    b &\geq \sum_{j = 1}^n a_jx_j 
    = \sum_{j \in J}a_jx_j + \sum_{j \not\in J} a_jx_j 
    = \sum_{j \in J}-|a_j|x_j + \sum_{j \not\in J} |a_j|x_j \iff
    \\
    b + \sum_{j \in J}|a_j| & \geq \sum_{j \in J}-|a_j|x_j + \sum_{j \not\in J} |a_j|x_j + \sum_{j \in J} |a_j| = \sum_{j \in J}|a_j|(-x_j + 1) + \sum_{j \not\in J} |a_j|x_j = \sum_{j = 1 }^n |a_j|(Ax + \is{J})_j 
\end{align*}
So \[x \in \{x : a^\top x \leq b\} \iff Ax +\is{J} \in \{x : |a|^\top x \leq b + \sum_{\ind \in [n] : a_\ind < 0} |a_\ind|\}\] and we can conclude that $R =Q$ and $\vol(R) = \vol(Q) = \vol(P)$.
\end{proof}
\begin{proof}[Proof of Proposition \ref{single-int-knap}]
    We can construct a one-to-one correspondence between points in $Z$ and $S$ as follows:
\\
    Consider $x^*\in Z$. Let $\bar{x} \in \{0,1\}^{n \times \log u}$ such that $\{\bar{x}_{i1},\bar{x}_{i2},\dots, \bar{x}_{i(\log u)}\}$ is the unique binary representation of $x^*_i$: $\sum_{j = 1}^{\log u } 2^{j-1}\bar{x}_{ij} = x^*_i$. It follows that 
    \[ \sum_{i = 1}^{n }\sum_{j = 1}^{\log u} 2^{j-1} a_i\bar{x}_{ij} = \sum_{i = 1}^{n } a_i\sum_{j = 1}^{\log u} 2^{j-1}\bar{x}_{ij} = \sum_{i = 1}^{n } a_i x^*_{i} \leq ub \]
and thus $\bar{x} \in S$.

Now consider $\bar{x} \in S$. 
Let $x^*$ be such that $x^*_i = \sum_{j = 1}^{\log u } 2^{j-1}\bar{x}_{ij}$. 
Note that $0 \leq \bar{x}_{ij} \leq 1$, and thus \[0 \leq \sum_{j = 1}^{\log u } 2^{j-1}\bar{x}_{ij} \leq \sum_{j = 1}^{\log u} 2^{j-1} = 2^{\log u - 1 + 1} - 1 = u - 1\]
Thus $x^* \in  \{0,\dots,u-1\}^n$. It also follows that 
\[ \sum_{i = 1}^{n } a_i x^*_{i} = \sum_{i = 1}^{n } a_i\sum_{j = 1}^{\log u } 2^{j-1}\bar{x}_{ij} = \sum_{i = 1}^{n }\sum_{j = 1}^{\log u} 2^{j-1} a_i\bar{x}_{ij}  \leq ub \]
and thus $x^* \in Z$. We can now conclude that $|Z| = |S|$.
\end{proof}

\subsection{Truncating Convex Constraint}\label{sec: single-conv-v1}
Consider a hypercube $[0,1]^n$ and an intersecting constraint $f(x) \le b$, where $f(x) = \sum_{j = 1}^n f_j(x_j)$, and each function $f_j : \R_{\geq 0} \to \R_{\geq 0}$ is nondecreasing, convex, and rational on $[0,1]$. Let $Q$ be the intersection of the unit cube with this constraint: $Q = \{x \in [0,1]^n : f(x) \leq b\}$.

\begin{proof}[Proof of Theorem \ref{thm:one-convex}]
First we use theorem \ref{find-l} to compute $\ell'$ such that $\frac12\ell(Q) \leq \ell'\leq \ell(Q)$. Let
\[u = \left\lceil\frac{9n^{2.5}}{\eps\ell'}\right\rceil,\qquad  g_{j}(x_j) = f_{j}({x_j}/u),\qquad g(x) = \sum_{j = 1}^n g_{j}(x_j).\]

Note that ${9n^{2.5}}/{\eps\ell(Q) \leq u \leq 18n^{2.5}}/{\eps\ell(Q)} + 1$.

Our algorithm then transforms convex body $Q$ to achieve $Q_u = [0,u]^n \cap \{x \in \R^n : g(x) \leq b\}$. We can do this through multiplying $Q$ with transformation matrix $u I_n$, thus $\vol(Q_u) = u^n\vol(Q)$.

Consider the integer points in $Q_u':= \{x \in [0,u)^n : g(x) \leq b \}$: \(Z = \Z^n \cap [0,u)^n \cap \{x: g(x) \leq b\} \)
Lemma \ref{int-pts-vol} tells us that given an $\frac{\eps}9$-approximation $Z'$ for $|Z|$, then $Z'$ serves as an $\eps$-approximation for $\vol(Q_u') =\vol(\cl(Q_u')) = \vol(Q_u)$.

Note that the number of integer points in $Q_u'$ is exactly the number of integer solutions $x \in \{0,1,\dots, u-1\}$ under the constraint $\sum_{j = 1}^n g_j(x_j) \leq b$.
We can now apply Theorem \ref{single-conv-ROBP} to give us a $\frac\eps9$ approximation of the integer points. This FPTAS has runtime $O(n^5 (\log u)^3 (L + \log n) / \del^2)$ where $\del = \frac\eps9$.
Recall that $u \leq 18n^{2.5}/\eps\ell + 1$ where $\ell = \ell(Q)$ is the smallest axis-intercept of $Q$. As seen in the proof of theorem \ref{find-l}, $\vol(Q) \leq \ell$ and $L_o \geq \log\frac1{\vol(Q)}$, and thus $\log(u) = O(\log(\frac{n}{\eps}) + \log (\frac1\ell)) = O(\log(\frac{n}{\eps}) + L_o)$.
Thus we perform
$O(n^5 (L_o+\log(\frac{n}\eps))^3 (L + \log n) / \eps^2)$ arithmetic operations. 
In both algorithms used, FIND\_INTERCEPT and ROUND\_ROBP, we operate on at most $O(L\log(\frac{n}{\eps}) + \log n + LL_o)$-bit numbers. 

We explicitly state the algorithm below:
\begin{tcolorbox}[colback=gray!10, colframe=gray!60]
\textbf{FPTAS for a Single Truncating Convex Constraint}

\textbf{ Input:} $f_{j} : \R \to \R_{\geq 0}$ for $j \in [n]$, $b\in \R_{\geq 0}$, $\eps >0$
\begin{enumerate}
    \item Call FIND\_INTERCEPT($\{x \in [0,1]^n : \sum_{j = 1}^n f_j(x_j) \leq b\}$), obtain $\ell'$.
    \item Compute $u \from \lceil{9n^{2.5}}/{\eps\ell'}\rceil$
    \item $Z \from \{x \in \Z^n : \sum_{j = 1}^nf_j( \frac{x_j}u) \leq b, x \in [0,u-1]^n\}$
    \item Call ROUND\_ROBP($\eps/9$) on set $Z$, obtain $Z'$ such that $|Z| \leq Z' \leq (1 + \eps/9)|Z|$.
    \item Return $\frac{Z'}{u^n}$
\end{enumerate}
\end{tcolorbox}
\end{proof}
\section{Hypercube Clipped by Multiple Constraints}
In this section we prove Theorem \ref{thm:multi-conv}. 

Let $f_i(x) = \sum_{j = 1}^n f_{ij}(x_j)$ where functions $f_{ij} : \R_{\geq 0} \to \R_{\geq 0}$ are nondecreasing and convex on $[0,1]$. Let 
\[
Q = \{x \in [0,1]^n: f_{i}(x) \leq b_i \quad \forall i \in [k]\}.
\]

As before, our algorithm first dilates our convex body $Q$ by some factor $u$. We then approximately count the number of integer points in this new convex body using an FPTAS developed in \cite{DBLP:journals/corr/abs-1008-3187}. This algorithm constructs small width Read-Once Branching Programs (ROBPs) to approximately count the number of points in such a set. We prove that this provides an approximate number, and thus volume, of unit cubes whose smallest coordinate is an integer point $x$ in the dilated body. We then show that an $\frac\eps9$ approximation of the volume of these cubes is an $\eps$ factor approximation of the volume of our body. 

\begin{proof}[Proof of Theorem \ref{thm:multi-conv}]
First we use theorem \ref{find-l} to compute $\ell'$ such that $\frac12\ell(Q) \leq \ell'\leq \ell(Q)$. Let \[u = \frac{9n^{2.5}}{\eps\ell'}, \quad g_{ij}(x_j) = f_{ij}(\frac{x_j}u), \quad g_i(x) = \sum_{j = 1}^n g_{ij}(x_j).\]
Note that ${9n^{2.5}}/{\eps\ell(Q) \leq u \leq 18n^{2.5}}/{\eps\ell(Q)}$.

    Our algorithm then transforms convex body $Q = [0,1]^n \cap \{x : f_i( x) \leq b_i$ $\forall i \in [k]\}$ to 
    $Q_u = [0,u]^n \cap \{x : g_i( x) \leq b_i$ $\forall i \in [k]\}$ through dilating by a factor of $u$. We do this by multiplying with transformation matrix $uI_n$, thus $\vol(Q_u) = u^n \vol(Q)$.

    Consider the integer points in $Q_u':= \{x \in [0,u)^n : g_i( x) \leq b_i$  $\forall i \in [k] \}$: 
    \[ Z = \Z^n \cap [0,u)^n \cap \{ x : g_i(x) \leq b_i \quad \forall i \in [k]\} = \Z^n \cap Q_u' \]

    Lemma \ref{int-pts-vol} tells us that given an $\frac{\eps}9$-approximation $Z'$ for $|Z|$, then $Z'$ serves as an $\eps$-approximation for $\vol(Q_u') =\vol(\cl(Q_u')) = \vol(Q_u)$.

It remains to find an FPTAS that counts the number of integer points in $Q_u'$. This is exactly the number of integer solutions $x \in \{0,1,\dots, u-1\}$ under the constraints $\sum_{j = 1}^n g_{ij}(x) \leq b_i\,\,\, \forall i \in [k]$.

We can now use the adapted FPTAS developed in \cite{DBLP:journals/corr/abs-1008-3187} to find the above. 
Theorem \ref{multi-ROBP} gives us such an FPTAS with runtime 
$O(n^{O(k^2)}(\log u/\eps)^{O(k)}(L + \log n))$.
Recall that $u \leq 18n^{2.5}/\eps\ell$ where $\ell = \ell(Q)$ is the smallest axis-intercept of $Q$. As seen in the proof of theorem \ref{find-l}, $\vol(Q) \leq \ell$ and $L_o \geq \log\frac1{\vol(Q)}$, and thus $\log(u) = O(\log(\frac{n}{\eps}) + \log (\frac1\ell)) = O(\log(\frac{n}{\eps}) + L_o)$.
Thus we use 
$n^{O(k^2)}(\log \frac{n}\eps + L_o)^{O(k)}(L + \log n)/\eps^{O(k)}$ arithmetic operations on $O(L(\log (n/\eps) + L_o))$-bit numbers.

We explicitly state the algorithm below:
\begin{tcolorbox}[colback=gray!10, colframe=gray!60]
\textbf{FPTAS for Multiple Truncating Convex Constraints}

\textbf{ Input:} $f_{ij} : \R \to \R_{\geq 0}$ for $i \in [k]$, $j \in [n]$, $b\in \R^k_{\geq 0}$, $\eps >0$
\begin{enumerate}
    \item Call FIND\_INTERCEPT($\{x \in [0,1]^n : \sum_{j = 1}^n f_{ij}(x_j) \leq b_i, \, \,\, \forall i \in [k]\}$), obtain $\ell'$.
    \item Compute $u \from \lceil{9n^{2.5}}/{\eps\ell'}\rceil$
    \item $Z_i \from \{x \in \Z^n : \sum_{j = 1}^nf_{ij}( \frac{x_j}u) \leq b_i\}, x \in [0,u-1]^n\},\quad Z \from \cap_{i = 1}^k Z_i$
    \item Call ROUND\_ROBPS($\eps/9$) on set $Z$, obtain $Z'$ such that $|Z| \leq Z' \leq (1 + \eps/9)|Z|$.
    \item Return $\frac{Z'}{u^n}$
\end{enumerate}
\end{tcolorbox}
\end{proof}

\section{Hardness of Approximately Counting solutions to Two Linear Inequalities}
\begin{proof}[Proof of Theorem \ref{thm: hardness}]
We prove that this problem is hard by reducing to a variant of the Subset-Sum Problem.
Consider a set $S = \{s_1, \dots, s_n\}$, where each $s_i \in \Z$, and target $T = 0$. The problem of deciding if there is a non-empty subset $I \seq [n]$ such that $\sum_{i \in I} s_i = 0$ is NP-hard.

This problem has a solution if and only if there exists an $x \in \{0,1\}^n \m \0$ such that $\sum_{i = 1}^n s_ix_i = 0$. Here $x_i$ can be though of as the indicator variable for whether or not $i \in I$.

Note that the sum $\sum_{i = 1}^n s_ix_i$ always has integral value for $x \in \{0,1\}^n$. Thus, if $\sum_{i = 1}^n s_i x_i \neq 0$, it must be true that $\sum_{i = 1}^n s_i x_i \geq 1$ or $\sum_{i = 1}^n s_i x_i \leq -1$. Similarly, any solution $x \in \{0,1\}^n$ satisfying $-1 <\sum_{i = 1}^n s_i x_i < 1$ must have sum exactly equal to $0$. 

Let $a\in \R^n$ be the vector such that $a_{i} = 2s_i$ and let $b = 1$.
Consider set \[Z = \{x \in \{0,1\}^n: -b \leq a^\top x \leq b \}.\]
There exists solution $x \in Z$ with $x \neq \0$ if and only if there is $x \in \{0,1\}^n\m \0$ such that $-1 \leq \sum_{i = 1}^n a_ix_i \leq 1$, or equivalently $-1< -\frac12 \leq \sum_{i = 1}^n s_ix_i \leq \frac12 < 1$. Thus, determining if there is a non-zero point in $Z$ is as hard as solving the subset sum variant. 
\end{proof}

\paragraph{Acknowledgement} The author is especially grateful to Santosh Vempala for introducing this problem to them, and for their guidance, encouragement, and many insightful discussions. The author also thanks Thiago Oliveira for their helpful comments and suggestions. The author also used ChatGPT (OpenAI) for proofreading and language editing during preparation of the manuscript.

\pagebreak
\bibliographystyle{alpha}
\bibliography{citations,acg}
\pagebreak
\appendix
\section{Appendix} In this section we provide background information and proofs of theorems \ref{single-conv-ROBP}, \ref{Dyer}, and \ref{multi-ROBP}.
\subsection{ROBP Preliminaries}
In this section we provide background information relevant to the proofs of \ref{single-conv-ROBP} and \ref{multi-ROBP}.
First, we note some definitions regarding Read-Once Branching Programs (ROBPs):
\begin{definition}[ROBPs]\label{robp}
For $u=(u_1,\ldots,u_n)\in\mathbb{Z}_+^n$, $S,T\in\mathbb{Z}_+$, an $(S,u,T)$-ROBP, $M$, is a layered multi-graph with a layer for each $ \el \in \{0,1,\dots, T\}$, and at most $S$ states in each layer. The first layer has a single (start) vertex $s$, and each vertex in the last layer is labeled accepting or rejecting. Each vertex $v$ in layer $\el -1$ has exactly $u_{\el}+1$ edges, labeled $\{0,1,\ldots,u_{\el}\}$, to layer $\el$.
\end{definition}
We also introduce the following notation
\begin{itemize}
    \item $\bm{L(M,\el)}$ denotes the set of vertices in layer $\el$ of $M$
    \item For a string $z \in \{0,1,\dots,u_{\el+1}\} \times \dots \times \{0,1,\dots,u_j\}$ and a vertex $v \in L(M,\el)$, $\bm{M(v,z)}$ denotes the state in layer $j$ reached by starting from $v$ and following edges labeled $z$.
    \item For $z \in \{0,1,\dots,u_1\} \times \dots \times \{0,1,\dots,u_n\}$, let $\bm{M(z)} = 1$ if $M(s, z)$ is an accept state, or \enquote{accepting}, and $\bm{M(z)} = 0$ otherwise.
    \item For $v \in L(M,\el)$, $\bm{A_M(v)} = \{z \in \{0,1,\dots,u_{\el + 1}\} \times \dots \times \{0,1,\dots,u_n\} : M(v,z)$ is accepting $ \}$ 
    \item For $v \in L(M,\el)$, $\bm{P_M(v)} = \frac{|A_M(v)|}{\prod_{j = \el}^n \big(u_j+1\big)}$ is the fraction of suffixes $z \in \{0,1,\dots,u_{\el + 1}\} \times \dots \times \{0,1,\dots,u_n\}$ that lead to an accepting state. This is equivalently the probability that a suffix sampled uniformly at random leads to an accepting state.
    \item The \textbf{width} of layer $\el$ of $M$ is exactly $|L(M, \el)|$, the number of states in the layer. Input $S$ is referred to as the \textbf{width} of $M$, and it upper bounds the width of each layer. 
\end{itemize}
We use this notation to define an important class of ROBPs, introduced in \cite{DBLP:journals/corr/abs-0910-4122} and \cite{DBLP:journals/corr/abs-1008-3187}, which will be the crux of our proofs.
\begin{definition}[Interval ROBPs]\label{interval-robp}
An $(S,u,T)$-interval ROBP, $M$, is an ROBP such that
there exists a total order $\prec$ on the vertices of layer $L(M,\el)$ such that, for $w,v \in L(M,\el)$ with $w \prec v$, then $A_M(w) \seq A_M(v)$. In addition, vertices $v$ in layer $\el -1$ have edges to layer $\el$ labeled $\{0,1,\ldots,u_{\el}\}$ that respect the ordering $\prec$: if $M(v,k)$ denotes the $k$th neighbor of $v$ for $k\le u_i$, then
$
M(v,u_{\el})\preceq M(v,u_{\el}-1)\preceq \cdots \preceq M(v,0).
$

An interval ROBP defines a function
$
M:\{0,\dots,u_1\}\times\{0,\dots,u_2\}\times\cdots\times\{0,\dots,u_n\}\to\{0,1\}
$
where on input $x$, we begin at the start vertex $s$ and output the label of the final vertex reached when traversing edges of $M$ according to labels in $x$.
\end{definition}

Note that given an $(S,u,T)$-interval ROBP, $M$, and a vertex $v\in L(M,\el-1)$, the edges out of $v$ can be described succinctly by a subset of at most $S$ edges irrespective of how large $u$ is. For each $w \in L(M,\el)$, let $
E(v,w)=\{k \in \{0, 1, \dots, u_{\el}\} : M(v,k)=w\}
$
which is the set of edges from $v$ to $w$. Then $E(v,w)$ is an interval, meaning $
E(v,w)=\{\ell_{v,w},\ldots,r_{v,w}\}
$
for some $\ell_{v,w},r_{v,w}\in\mathbb{Z}_+$. 

Thus, to describe $E(v,w)$ we only need to know $\ell_{v,w}$ and $r_{v,w}$. This allows us to succinctly describe an interval ROBP ${M}$ by storing just the end points of the edge sets $E(v,w)$ for $v,w\in {M}$. This bounds the number of edge sets by the number of vertices in the subsequent layer, which is at most width $S$.
\subsection{Read Once Branching Program for a single read-once convex constraint}\label{sec:sinlge-robp}
\begin{proof}[Proof of Theorem \ref{single-conv-ROBP}]
We first state the related theorem from \cite{DBLP:journals/corr/abs-1008-3187}:
\begin{theorem*}[integer knapsack]
    Given a knapsack instance $\mathrm{KNAP}(a,b,u) = \{x \in \Z : \sum_{j = 1}^n a_j x_j \leq b, 0 \leq x_j \leq u_j \quad \forall j \in [n] \}$ with weight
$W = \sum_i a_i u_i + b$, $U = \max_j u_j$ and $\eps > 0$,
there is a deterministic
$
O\!\left(n^5 (\log U)^2 (\log W) / \del ^2\right)
$
algorithm that computes an $\del$-relative error approximation for
$|\mathrm{KNAP}(a,b,u)|$.
\end{theorem*}
In our theorem, we substitute $\KNAP(a,b,u)$ with $Z(f,b,u) = \{x \in \Z : \sum_{j = 1}^nf_j(x_j) \leq b, 0 \leq x_j \leq u-1 \quad \forall j \in [n] \}$ where each $f_j$ is nonnegative, nondecreasing, convex, and rational. These functions obey all the properties used in the proof of the theorem of \cite{DBLP:journals/corr/abs-1008-3187}.

We now provide a detailed replica of the proof with the appropriate substitutions made.

We will construct an approximate branching program for our set $Z(f,b,u)$. We will do so by selecting a subset of vertices in the exact ROBP of $Z(f,b,u)$ to remain, and alter the edge sets accordingly. Note that throughout the proof we will assume that, without loss of generality, $f_j(0) = 0$ for all $j \in [n]$, which follows from observation \ref{obs}.

Let $M$ denote the exact branching program for $\sum_{j = 1}^n f_j(x_j) \leq b$, which consists of $n+1$ layers numbered $0$ to $n$. 
Layer $0$ has single start state $s$. For $\el \leq n$, layer $L(M,\el)$ has a state for every partial sum $\sum_{j = 1}^{\el} f_j(x_j)$ where each $x_j \in \{0,1,\dots,u-1\}$. 
Given a vertex $v$ in layer $\el$, and suffix $x_{\el+1} \in \{0,1, \dots, u-1\}$, the $x_{\el+1}^{th}$ neighbor of $v$ is $M(v,x_{\el+1}) = v + f_{\el + 1}(x_{\el+1})$. A vertex $v \in L(M,n)$ is accepting if partial sum $v \leq b$, and rejecting otherwise.

Note that this program is clearly read once, as each layer corresponds to exactly one variable. It is also interval with total ordering $\prec$ defined as $>$. We can see this because for any $v,w\in L(M,\el)$ with $v > w$ and any suffix $z \in \{0,1, \dots, u-1\}^{n - \el}$, we have $v + \sum_{j = \el + 1}^n f_j(z_j) > w + \sum_{j = \el + 1}^n f_j(z_j)$. Thus if $z \in A_M(v)$, or $v + \sum_{j = \el + 1}^n f_j(z_j) \leq b$, then clearly $w + \sum_{j = \el + 1}^n f_j(z_j) \leq b$, thus $z \in A_M(w)$.
Consequently, $A_M(v) \seq A_M(w)$. Also, for any vertex $v \in L(M,\el)$ and pair of suffixes $0 \leq \ell < r \leq u-1$, we have $M(v,\ell) = v + f_{\el+1}(\ell) \leq v + f_{\el+1}(r) = M(v,r)$ since each $f_{\el+1}$ is nondecreasing. Thus, $M(v,\ell) \succeq M(v,r)$.

By the construction of the states, the width of layer $L(M,\el)$ is bounded by the total number of partial sums between our minimum and maximum possibilities: $[0,\sum_{j  = 1}^{\el}f_j(u-1)]$. There can be up to $u^n$ vertices in each layer. Since this is exponential, we will construct a new ROBP with smaller width. A small width program will allow us to evaluate the number of accepting solutions in polynomial time.

Note that since we do not have access to which values $\sum_{j = 1}^{\el}f_j(x_j)$ actually exist in layer $\el$, we assume that $L(M,\el)$ has \emph{all} values in $[0,\sum_{j = 1}^{\el}f_j(u-1)]$ that have encoding length $L + \log u + \log n$.

Our algorithm will construct a series of interval ROBPs $\M_n=M,\M_{n-1},\ldots,\M_0$. Our final ROBP $\M := \M_0$ will have polynomially bounded width. We obtain $\M_j$ from $\M_{j+1}$ by \emph{rounding} the states in $L(\M_{j+1},j+1)$. More precisely, we set
$
L(\M_j,j+1)=\{\beta_1,\ldots,\beta_{N+1}\}\subseteq L(\M_{j+1},j+1) 
$
where the $\beta_i$’s are defined as follows: let $\beta_1= 0$, which is the minimum value state in layer $j + 1$, and equivalently the one with the most accepting suffixes. Given $\beta_i$, let
\[  \beta_{i+1}=\min\{v: v> \beta_i \text{ and } 0< P_{\M_{j+1}}(v) < P_{\M_{j+1}}(\beta_i)/(1+\eta)\}, \tag{A.1}
\]
for $\eta$ to be set later.
Now consider $\beta_N$, the last vertex assigned. Let $\beta_{N + 1} = \min \{v : P_{\M_{j + 1}}(v) = 0 \}$. Note that $\beta_1 \succ \beta_2 \succ \dots \succ \beta_{N} \succ \beta_{N + 1} $ as the probabilities decrease.
Also note that $N \leq n(\log u)/\eta$ as $P_{\M_j}(\beta_1)\leq 1$ and $P_{\M_j}(\beta_N)\ge u^{-n}$. Let $L' = L \lceil\log u\rceil + \log n$, and let
\[
I_1=[\beta_1,\beta_2- 1/2^{2L'}], I_2=[\beta_2,\beta_3-1/2^{2L'}],I_3=[\beta_3,\beta_4-1/2^{2L'}],\;\ldots,\; I_{N+ 1}=[\beta_{N+1},\sum_{\el = 1}^{j + 1}f_\el(u -1)].
\]
We refer to $\beta_1, \dots, \beta_{N+1}$ as breakpoints since they break our original vertices into intervals. Next we redirect the transitions going from level $j$ to level $j+1$. If we have an edge labeled $z\in\{0,\ldots,u-1\}$ entering a vertex $v\in I_i$, then we redirect the edge to vertex $\beta_i$. The redirection will be done implicitly: for any vertex $v$ in level $j$ and a breakpoint $\beta_i$, we only compute and store the end points of the interval $
E(v,\beta_i)=\{k \in \{0, 1, \dots, u-1\} : \M_j(v,k)=\beta_i\}.
$

Our branching programs have the following approximating properties:

\begin{claim}[Corresponding to the first half of Lemma A.2. of \cite{DBLP:journals/corr/abs-1008-3187}]\label{A.2.1}
    For $v \in L(\M_j,\el)$ and $0 \leq \ell < r \leq u-1$, we have $\M_{j}(v,\ell) \leq \M_j(v,r)$.
\end{claim}
\begin{proof}
    Note that this claim is only interesting if $\el > j$, otherwise all vertices $v, \M_{j}(v,\ell), \M_j(v,r)$ have the same structure in the original interval ROBP $M$ and thus the claim holds by the interval property. 
    Thus we can assume that layer $\el + 1$ is rounded and consists of breakpoints. Take $\M_{j}(v,\ell) = \beta_{x}$ and $\M_j(v,r) = \beta_y$.

    Assume towards contradiction that $\beta_x > \beta_y$.
    
    When creating edge labeled $\ell$ from vertex $v$, we found the value of $v_{\ell} = v + f_{\el + 1}(\ell)$, and similarly for edge labeled $r$, we found the value of $v_r = v + f_{\el + 1}(r)$. We then found 
    breakpoints such that $\beta_{x}\leq v_{\ell} < \beta_{x+1}$ and $\beta_y\leq v_r < \beta_{y+1}$. 
    Since $\ell < r$ and $f_{\el + 1}$ is nondecreasing, we know that $v_{\ell} = v + f_{\el+1}(\ell) \leq  v + f_{\el+1}(r) = v_r$. Thus we see that $\beta_{y+1} > v_r \geq v_{\ell} \geq \beta_x > \beta_y$, which implies that there is a breakpoint ($\beta_x)$ between (and not equal to) consecutive breakpoints $\beta_{y+1}$, $\beta_y$ which is a contradiction, and so the claim holds.
\end{proof}
\begin{claim}[Corresponding to the second half of Lemma A.2. of \cite{DBLP:journals/corr/abs-1008-3187}]\label{A.2.2}
    Let $v,v' \in L(\M_j,\el)$, $v \leq v'$. For suffix $z$, $\M_j(v,z) \leq \M_j(v',z)$.
\end{claim}
\begin{proof}
    Note that we need only prove this claim for suffixes of length one, as the property for any longer suffix will follow inductively.
    
    Also note as in the previous theorem, this claim is only interesting if $\el > j$, otherwise vertices $v, v',\M_j(v,z), \M_j(v',z)$ have the same structure in the original interval ROBP $M$. Thus we can assume that layer $\el + 1$ is rounded and consists of breakpoints.
    
    Consider two vertices $v,v'$ in layer $\el$ of our partially rounded ROBP such that $v \leq v'$. We aim to show that suffixes maintain the ordering. Take $\M_{j}(v,z) = \beta_{x}$ and $\M_j(v',z) = \beta_y$.
    
    Assume towards contradiction hat $\beta_{x} > \beta_{y}$. 
    Take $v_{z} = v + f_{\el + 1}(z)$, $v'_{z} = v' + f_{\el + 1}(z)$ and note that since $v \leq v'$, we have $v_z \leq v'_z$.
    Similarly to the previous claim, we see that $\beta_{y+1}> v'_z \geq v_z \geq \beta_x > \beta_y$, which implies that there is a breakpoint ($\beta_x)$ between (and not equal to) consecutive breakpoints $\beta_{y+1}$, $\beta_y$ which is a contradiction, and so the claim holds. 
\end{proof}
By inductively applying the previous two claims to each layer, we see that each $\M_j$ remains an interval ROBP respecting the same ordering as $M$. 
We now analyze how the probability of acceptance (and thus the total number of solutions) in our new ROBP $\M$ compares to that of the original ROBP $M$. 
\begin{claim}[Corresponding to the first half Lemma A.3. of \cite{DBLP:journals/corr/abs-1008-3187}]\label{A.3.1}
For $v\in \M_j$, we have $A_{\M_{j+1}}(v)\subseteq A_{\M_j}(v)$. 
\end{claim}
\begin{proof}
    Note that for any vertex $v\in L(\M_j,\el)$ where $\el > j$, both ROBPs make identical transitions starting from $v$, and thus $A_{\M_{j+ 1}}(v) = A_{\M_j}(v)$. 
    
    Let $\el = j$. 
    Let $\beta_z = \M_j(v,z)$ be the child of $v$ along edge $z \in \{0,1,\dots,u-1\}$. Note that since $\M_j$ is obtained from $\M_{j + 1}$ by rounding layer $j + 1$, these children are breakpoints, and the \enquote{true} children of $v$ are $v_z = M_{j + 1}(v,z)$ in layer $i + 1$ of $\M_{j + 1}$ such that $v_z \geq \beta_z$. 
    
    Since each $\M_j$ remains an interval ROBP and the structure below $\beta_z$ is the same in $\M_j$ and $\M_{j + 1}$, we have $A_{\M_{j+1}}(v_z) \seq A_{\M_{j+1}}(\beta_z)=A_{\M_{j}}(\beta_z)$. Thus, taking the union over each suffix $z$ placed ahead of each string in $A_{\M_{j+1}}(v_z)$ and $A_{\M_{j}}(\beta_z)$, we maintain this containment, and $A_{\M_{j+1}}(v) \seq A_{\M_j} (v)$.
    
    Now for any $\el < j$, the claim follows because $\M_j$ and $\M_{j + 1}$ make identical transitions up to layer $j$. Thus we can take the union over suffixes used up to layer $j$ and achieve the same property. 
\end{proof}  
\begin{claim}[Corresponding to the second half of Lemma A.3. of \cite{DBLP:journals/corr/abs-1008-3187}]\label{A.3.2}
For any $v\in L(\M_j,\el)$ where $\el\leq j$, we have
$
P_M(v)\leq P_{\M_j}(v)\leq P_M(v)(1+\eta)^{\,n-j}.
$
\end{claim}
\begin{proof}
    Note that since $A_{\M_{j+1}}(v) \seq A_{\M_j}(v)$, it is clear that $P_{\M_{j + 1}}(v) \leq P_{\M_j}(v)$ for all $v$, and thus $P_{M}(v) = P_{\M_n}(v) \leq P_{\M_{n-1}}(v)\leq \dots \leq P_{\M_j}(v)$.

    We will show that for every $j<n$ and $v\in \M_j$, $P_{\M_j}(v) \leq P_{\M_{j+1}}(v)(1+\eta)$, and thus $P_{\M_j}(v) \leq P_{\M_{n}}(v)(1+\eta)^{n-j}$ as needed.

Let $v \in L(\M_j,\el)$.
The above is trivial when $\el \geq j+1$, since
$A_{\M_j}(v)=A_{\M_{j+1}}(v)$ for such $v$.
Indeed, it suffices to consider the case when $\el=j$, since for $\el<j$,
$\M_j$ and $M$ are identical up to layer $j$.
Hence we can express both $P_{\M_{j+1}}(v)$ and $P_{\M_j}(v)$ as the same convex
combination of acceptance probabilities of vertices in layer $j$.

Let $\el=j$. Fix a vertex $v \in L(\M_j,j)$.
For any $z\in\{0,1, \dots, u-1\}$, let $v_z = \M_{j + 1}(v,z)$ be the \enquote{true} children of $v$ and let 
$\beta_{x_z}=\M_j(v,z)$
breakpoint children. 
Recall that $\M_j$ is obtained from $\M_{j+1}$ by rounding layer $j$, and $\beta_{x_z} \leq v_z < \beta_{x_z + 1}$. 

Note that $P_{\M_{j+1}}(v_z) \geq P_{\M_{j+1}}(\beta_{x_z})/(1+\eta)$ or else $v_z$ could form a new breakpoint, thus $P_{\M_{j+1}}(\beta_{x_z}) \leq (1+\eta) P_{\M_{j+1}}(v_z)$. Also, recall that $\beta_{x_z}$ exists with the same structure in both $\M_{j + 1}$ and $\M_j$. Thus, $P_{\M_{j}}(\beta_{x_z}) = P_{\M_{j+1}}(\beta_{x_z})$, and
\[
P_{\M_j}(v)
= \tfrac1u\left(\sum_{z = 0}^{u-1 }P_{\M_j}(\beta_{x_z})\right)
\le (1+\eta)\frac1u{\sum_{z = 0}^{u-1 }P_{\M_{j+1}}(v_z)}
= (1+\eta)P_{\M_{j+1}}(v).
\]
\end{proof}
We next show that $\M$ can be constructed efficiently.

\begin{claim}[Corresponding to Lemma A.4 of \cite{DBLP:journals/corr/abs-1008-3187}]\label{A.4}
    Each vertex $v_j \in L(\M_j,j+1)$ can be computed in time $O(n(\log u)(\log b + L\log u + \log n )/\eta)$.
\end{claim}
\begin{proof}
The proof is by induction: we maintain the invariant that for every $i$, we know the vertices
$\beta_z$ of $L(\M_j,j+1)$ as well as their acceptance probabilities $P_{\M_j}(\cdot)$, and the encoding length of each $\beta_z$ is at most $L' = L 
\lceil\log u\rceil + \log n$.

Note that if $j = n-1$, we round layer $n$ by creating breakpoints $\beta_1 = 0$ and $\beta_2 = b+1/2^{2L'}$. Clearly for any $v \in I_1 = [\beta_1, \beta_2 -1/2^{2L'}] = [0,  b+1/2^{2L'} - 1/2^{2L'}] = [0,b]$, it follows that $v$ accepts with probability $1$. Otherwise, $v \geq b + 1/2^{2L'}$, and thus accepts with probability $0$. 
For all other layers we choose $\beta_1 = 0$ to begin with.

Suppose we have the above invariant for $l>j$ and we have computed $\beta_1,\dots,\beta_i \in L(\M_j,j+1)$. In order to find the next breakpoint, we must have access to probabilities of vertices in layer $j + 2$.

We show that for a given $v \in L(\M_{j+1},j+1)$, $P_{\M_{j+1}}(v)$ can be computed in time
$O(n(\log u)/\eta)$.
Let $L(\M_{j+1},j+2)=\{\beta_1< \beta_2< \cdots < \beta_{N + 1}\}$ and
$E(v,\beta_i)=\{k \in \{0, 1, \dots, u-1\}  : \beta_i \leq v+f_{j+2}(k) \leq \beta_{i+1}\}$ where $\beta_{N + 2}$ is $v_W$ and thus
\[
P_{\M_{j+1}}(v)
=
\sum_{w\in L(\M_{j+1},j+2)}
\frac{|E(v,w)|}{u}\,
P_{\M_{j+1}}(w).
\]
Note that the minimum and maximum edge labels that describe $E(v,w)$ can be determined for each $w$ using binary search through possible edge labels.
Thus, we can compute $P_{\M_{j+1}}(v)$ in time $O(n(\log u)/\eta)$ as
$|L(\M_{j+1},j+2)| \le 1 + n(\log u)/\eta$.
To find $\beta_{i+1}$,
we can do binary search on values in $[0,b]$, computing their associated $P_{\M_{j+1}}(v)$ each time. This range is acceptable because
we require $0 < P_{\M_{j + 1}}(v)$, so we need only search through vertices in the range $[0,b]$ since any value larger than $b$ can never lead to an accepting suffix due to nonnegativity of each $f_j$. We will only consider breakpoints with bit size $L'$ which requires
$O(\log b + L') = O(\log b + L \log u + \log n)$ iterations. 
Finally notice that $P_{\M_{j + 1}}(b + 1/2^{2L'}) = 0$ and $P_{\M_{j + 1}}(b ) > 0$ because the suffix string of $0$s accepts: $b + \sum_{j = i + 2}^n f_j(0) = b + 0 \leq b$.
Thus we can set $\beta_{N + 1} = b + 1/2^{2L'}$.
Once we have computed $L(\M_j,j+1)$ we store these vertices and their probabilities of acceptance.
\end{proof}
It is important to note that computing breakpoints to $L'$ bit accuracy is sufficient. Observe that, in the true ROBP $M$, we only evaluate $f_{j}(x_j)$ for integers $x_j\in\{0,...,u-1\}$, which have encoding length at most $\lceil \log u \rceil$. Since $f_j$ is a rational function, it follows that the encoding length of $f_{j}(x_j)$ is at most $L\lceil \log u \rceil$. This means that each partial sum value $\sum_{j = 1}^nf_{j}(x_j)$ has encoding length at most $\log n + L\lceil \log u \rceil = L'$. Thus, all the \enquote{true} vertices of the ROBP are considered in out binary search.  

In addition, notice that we construct the last breakpoint, with $0$ probability, as $b + 1/2^{2L'}$. We use the following claim to verify that there is no true child strictly between $b$ and $b + 1/2^{2L'}$. 
\begin{claim}\label{encoding-gap}
    Given a value $v>0$ with encoding length at most $L'$, there is no positive rational number with encoding length at most $L'$ that lies strictly between $v$ and $v + 1/2^{2L'}$.
\end{claim}
\begin{proof}
Take the fraction representing $v$: $p/q$ such that $p,q \in \Z_{+}$. Consider another $L'$-bit rational number $y = s/t$ such that $y \neq v$. Note that $L'\geq \lceil \log p + 1\rceil + \lceil \log q + 1 \rceil \geq \log q$ and  $L'\geq \lceil \log s + 1\rceil + \lceil \log t + 1 \rceil \geq \log t$. Thus, $q\leq 2^{L'}$ and $t \leq 2^{L'}$. Thus, 
\[\left| \frac{p}{q} - \frac{s}t\right| = \left| \frac{pt - sq}{qt} \right| \geq \frac{1}{qt} \geq \frac{1}{2^{2L'}} \]
where the first inequality comes from the fact that $y \neq v$, and thus the difference is $>0$, but $p,t,s,q$ are integers. Thus, it cannot be true that two $L'$ bit rational inputs have distance less than $1/2^{2L'}$.
\end{proof}
This tells us that there are no true children with $0$ probability that we accidentally assigning to a breakpoint with nonzero probability.

Thus we can construct $\M = \M_0$ from $M$ in time $O(n^3(\log u)^2(\log b + L\log u + \log n)/\eta^2)$ where $\log b \leq L$ as it is an input.
We can now finish the proof of our counting result:

We set $\eta=\delta/2n$ and use the above arguments to construct the branching
program $\M$ and compute the value of $P_{\M}(s)$ where $s$ is the start state.
We now apply Lemma A.3 to conclude that
\[
P_M(s) \le P_{\M}(s) \le P_M(s)(1+\eta)^n \le (1+\delta)P_M(s)
\]
where the last inequality holds for small enough $\delta$.
Finally, note that the number of integer solutions is precisely
$u^nP_M(s)$.
Hence we output $u^nP_{M^0}(s)$.
 \end{proof}

\section{Dyer's Rounding for Multiple Sets}\label{sec:dyer}
\begin{proof}[Proof of theorem \ref{Dyer}]
Let $U_n = \{0,\dots,u_1\}\times \dots \times \{0,\dots,u_n\}$.
Consider set $Z = \cap_{i \in [k]} Z_i$ where $Z_i =\{ x \in U_n : \sum_{j = 1}^nf_{ij}(x_j) \leq b_i \}$ and each $f_{ij} : \R \rightarrow \R_{\geq 0}$ is nondecreasing, convex, and rational. 
We will assume without loss of generality that $f_{ij}(0) = 0$ for all $j \in [n]$, $i \in [k]$, which follows from observation \ref{obs}.

We expand on the results of \cite{10.1145/780542.780643} to find an approximate counting of $|Z|$. We create new sets $S_i = \{ x \in U_n : \sum_{j = 1}^nh_{ij}(x_j) \leq c_i \}$ where each $c_i$ is polynomially bounded, and the size of $S = \cap_{i \in [k]}S_i$ is bounded between the size of $Z$, and some multiple of the size of $Z$.

Specifically, we construct rounded set \[S = \bigcap_{i =1}^k S_i\text{, where }S_i = \{ x \in U_n : \sum_{j = 1} ^n h_{ij}( x_j) \leq 2n^2\}\] such that the rounded functions are \[h_{ij}(x_j) = \lfloor \frac{2n^2f_{ij}(x_j)}{b_i} \rfloor \quad \text{ for } 0 \leq x_j \leq u_j, j \in [n], i \in [k]\]
Observe that for any $x \in Z_i$, we have that 
\[\sum_{j = 1} ^n h_{ij}( x_j) = \sum_{j = 1} ^n \lfloor \frac{2n^2f_{ij}(x_j)}{b_i}\rfloor \leq \sum_{j = 1} ^n  \frac{2n^2f_{ij}(x_j)}{b_i} = \frac{2n^2}{b_i} \sum_{j = 1} ^n  f_{ij}(x_j) \leq \frac{2n^2}{b_i} b_i = 2n^2, \] and thus $x \in S_i$, and thus $Z_i \seq S_i$.

Now let $C_i = \{x : f_{ij}(x_j) \leq \frac{b_i}{n}\text{ } 
\forall j \in [n]\}$. Note that for all $x \in C_i$, it follows that $\sum_{j = 1}^n f_{ij} (x_j) \leq \sum_{j = 1}^n \frac{b_i}{n} = n\frac{b_i}{n} = b_i$, and thus $C_i \seq Z_i \seq S_i$

For $x \in S \m Z$, let $I(x) = \{ i : x \in S_i \m Z_i \}$. For each $i \in I(x)$, choose $p_i^x \in [n]$ such that $f_{i{p_i^x}}(x_{p_i^x}) > \frac{b_i}n$, which must exist, as otherwise $x \in C_i \seq Z_i$. Note $h_{i{p_i^x}}(x_{p_i^x}) = \lfloor \frac{2n^2f_{i{p_i^x}}(x_{p_i^x})}{b_i}\rfloor \geq \lfloor \frac{2n^2b_i /n}{b_i}\rfloor = \lfloor \frac{2n^2}{n}\rfloor = 2n$.

Construct $g(x): S \rightarrow U_n$ as follows: $g(x) = \begin{cases}
    x &\text{ if }x \in Z\\
    y &\text{ where }y_{p_i^x} = \lfloor{x_{p_i^x}}/2\rfloor  \,\forall i \in I(x)\text{, and $y_j = x_j$ otherwise}
\end{cases}$
By convex and nondecreasing properties of each $f$ it follows that $f_{ij}(\lfloor{x_{j}}/2\rfloor) \leq f_{ij}({x_j}/2) \leq \frac12 f_{ij}(x_j)$.
Thus, for all $x \not\in Z$, $y = g(x)$, and $i \in I(x)$, we have 
\begin{align*}
    \sum_{j = 1}^n f_{ij}(y_j) &= \frac{b_i}{2n^2}\sum_{j = 1}^n \frac{2n^2f_{ij}(y_j)}{b_i} \\
    &\leq \frac{b_i}{2n^2}\bigg(\sum_{j \neq {p_i^x}}( \frac{2n^2f_{ij}(x_j)}{b_i}) + \frac{2n^2f_{i{p_i^x}}(\lfloor {x_{p_i^x}}/2 \rfloor )}{b_i}\bigg) \\
    &\leq \frac{b_i}{2n^2}\bigg(\sum_{j \neq {p_i^x}} (h_{ij}(x_j) + 1) + \frac12 \cdot \frac{2n^2f_{i{p_i^x}} (x_{p_i^x})  }{b_i}\bigg) &\text{By convexity of $f_{ij}$}\\
    &\leq \frac{b_i}{2n^2}\bigg(\sum_{j \neq {p_i^x}} (h_{ij}(x_j)) + n-1 + \frac12 (h_{i{p_i^x}}(x_{p_i^x}) + 1)\bigg) \\
    &\leq \frac{b_i}{2n^2}\bigg(\sum_{j \neq {p_i^x}} (h_{ij}(x_j)) + n-1 + h_{i{p_i^x}}(x_{p_i^x}) - n + \frac12 \bigg) &\text{By bound on } h_{i{p_i^x}}(x_{p_i^x})\\
    &\leq  \frac{b_i}{2n^2}\bigg( 2n^2 - \frac12 \bigg) < b_i
\end{align*}
Thus, for every $x \in S$ and $i \in I(x)$, we have $g(x) \in Z_i$. For each $i \notin I(x)$, we have $x \in Z_i$ and $g(x) \leq x$, so by monotonicity $g(x) \in Z_i$ as well. Therefore $g(x) \in Z$ for every $x \in S$, and hence $g(S) \subseteq Z$. Since $g(y)=y$ for every $y \in Z$, we have $g(S)=Z$.

Now fix $y \in Z$ and consider $x \in g^{-1}(y)$. Let $P \subseteq [n]$ be the set of coordinates changed by $g$. Since at most one coordinate is chosen for each $i \in I(x)$, we have $|P| \leq k$. For every $j \notin P$, $x_j=y_j$, while for every $j \in P$,
$y_j=\left\lfloor \frac{x_j}{2}\right\rfloor$,
so there are at most two possible values of $x_j$, namely $2y_j$ and $2y_j+1$. Therefore
\[
|g^{-1}(y)|
\leq
\sum_{r=0}^k 2^r \binom{n}{r}
\leq
2(2n)^k.
\]
Since $g(S)=Z$ and $Z \subseteq S$, we conclude that
\[
|Z| \leq |S| = 
\sum_{y\in Z}|g^{-1}(y)| \leq 2(2n)^k|Z|.
\]

\end{proof}

\subsection{Read Once Branching Program for multiple read-once convex constraint}\label{sec:multi-robp}\label{ROBP}
In this section we prove theorem \ref{multi-ROBP}, which is an extension of work of \cite{DBLP:journals/corr/abs-1008-3187}, which gives us a proof of statement when $u = 2$. 

We will develop a deterministic approximate counting algorithm for set $Z = \cap_{i = 1}^kZ_i$, where $Z_i= \{x\in \{0,1,\dots, u-1\} : \sum_{j = 1}^n f_{ij}(x_j) \leq b_i\}$. Note that throughout the proof we will assume without loss of generality that $f_{ij}(0)=0$ for all $i\in[k],j\in[n]$, which follows from observation \ref{obs}.

We will again rely on Interval-ROBPs (\ref{interval-robp}). We introduce one more key concept, developed in \cite{10.1145/1132516.1132613}:
\begin{definition}[small-space sources, Kamp et al.]
    A width $w$ small-space source is described by a $(w,u,n)$-branching program $D$ with an additional probability distribution $p_v$ on the outgoing edges associated with vertices $v \in D$. Samples from the source are generated by taking a random walk on $D$ according to the $p_v$'s and outputting the labels of the edges traversed.
\end{definition}
We will often abuse notation and denote both the distribution generated by a small-space source and the source itself by $D$. Also, we will assume that the distribution $D$ is given to us explicitly as a small-space source. 
The following claim is directly from \cite{DBLP:journals/corr/abs-1008-3187}:

\begin{claim}[Equivalent to Claim 2.4 in \cite{DBLP:journals/corr/abs-1008-3187}]\label{2.4}
    Given a $(W,n)$-ROBP $M$, the uniform distribution over $M$’s accepting inputs, $\{x : M(x) = 1\}$ is a width $W$ small-space source.
\end{claim}

Let $D$ denote a small space generator of width at most $S$. We use the following notation:
\begin{itemize}
    \item For $A \subseteq U_n$ we use $\bm{D(A)}$ to denote the measure of $A$ under $D$.
    \item $\W^1,\ldots,\W^n$ are the sets of vertices in $D$, with $\W^{\el}$ being the $\el^{\text{th}}$ layer of $D$.
    \item For a vertex $w \in \W^{\el}$, we let $\bm{D^w}$ be the distribution over $\{0,\dots, u_{{\el}+1}\} \times \dots \times\{0,\dots,u_n\}$ induced by taking a random walk in $D$ starting from $w$.
    \item Given $v \in L(M,{\el})$ and $w \in \W^{\el}$, we let $\bm{P_{M,w}(v)}$ denote the probability of accepting if we start from $v$ and make transitions in $M$ according to a suffix sampled from distribution $D^w$.
\end{itemize}  

In this algorithm, our small space source will be constructed from the small width rounded knapsack instances built in \ref{Dyer}. 
Explicitly we will build $D$ by creating layers $0,1,\dots,n$. There will be a single start vertex $w_0$, and each vertex in layer $\el-1$ will have edges labeled $0,1,\dots, u-1$ to vertices in layer $\el $. Given prefix $x \in \{0,1,\dots,u-1\}^{\el}$, let $w_x$ be the vertex in layer $\el$ reached through traversing edge labels according $x$, starting from $w_0$.
For every such $x$, we have $w_x = (v_1, v_2, \dots v_k)$, which is the partial sum tuple for each constraint, i.e. $v_i = \sum_{j = 1}^{\el} h_{ij}(x_j)$. We can see that for $w_x = (v_1, v_2, \dots v_k)$ in layer $\el-1$, the $d^{th}$ child of $w_x$ has partial sum tuple $(v_1 + h_{1\el}(d), v_2+ h_{2\el}(d), \dots v_k+ h_{k\el}(d))$. 
Within each layer, we contract all vertices that have identical partial sum tuples.
Vertex $w = (v_1, \dots, v_k)$ in layer $n$ will be labeled accepting if $v_i \leq 2n^2$ for all $i \in [k]$, and will be labeled rejecting otherwise.

In each layer, we will delete all partial sum tuples $(v_1,\dots,v_k)$ where $v_i > 2n^2$ for some $i \in [k]$. This will not affect the probability distribution on the edges entering or leaving this vertex, as the nonnegativity of our functions guarantee it will lead to a rejecting vertex.

Recall that each function $h_{ij}(x_j)$ takes integer values, and each constraint has polynomial capacity $2n^2$. This implies that each partial sum can have value at most $2n^3$: there are at most $n$ items each with function value at most $2n^2$. If the function value exceeds this amount, there exists $d \in \{0,1,\dots,u\}$ such that $h_{i\el}(d) > 2n^2$, and thus we can simply delete all edges with labels $d, \dots, u_1$, as they surely lead to a failing vertex. Thus, the width of this ROBP is at most $(2n^3)^k$.

We can also bound the number of edges by this value through creating intervals: consider the edge set $E(v,w) = \{d \in \{0,1,\dots, u-1\} : D(v,d) = w)\}$. 
This is an interval because if any $\ell< m < r$ satisfies $h_{i\el}(\ell) = h_{i\el}(r)$, $\forall i$, then $h_{i\el}(\ell) = \lfloor \frac{2n^2f_{i\el}(\ell)}{b_i} \rfloor \leq \lfloor \frac{2n^2f_{i\el}(m)}{b_i} \rfloor \leq \lfloor \frac{2n^2f_{i\el}(r)}{b_i} \rfloor = h_{i\el}(r)$, and so $h_{i\el}(\ell) = h_{i\el}(m) = h_{i\el}(r)$. Thus to describe $E(v,w)$, we need only store $\ell_{vw}, r_{vw}$.

Lastly, we will compute probabilities on the edge intervals. Recall that vertices in layer $n$ are accepting (labeled $1$) if every element in the tuple is at most $2n^2$, and is rejecting (labeled $0$) otherwise. We will remove all rejecting vertices and edges to them, since they will have probability $0$. Any remaining edge from $v$ in layer $n-1$ to $w$ in layer $n$ represents $|E(v,w)|$ accepting children of $v$. The total number of accepting children of $v$, which we refer to as $A_D(v)$, will be the number of edges out of $v$, since all non-solution children have been deleted. To find this, we need only compute $\sum_{w \in L(D,n)}|E(v,w)| = \sum_{w \in L(D,n)} (r_{vw} - \ell_{vw} + 1)$. For any edge labels $d\in \{0,1,\dots,u-1\}$ (not deleted) we will have probability $p_v(d) = {1}/{|A_D(v)|}$. This clearly gives us the distribution over accepting solutions below each $v$ in layer $n-1$.

We will calculate values $p_v$ and $|A_D(v)|$ for vertices in layers above iteratively. We assume we have computed these values for layers $n, n-1, \dots, \el+1$. For any vertex $v$ in layer $\el$, we repeat a similar process. Any child $w \in L(D,\el + 1)$ of $v\in L(D, \el)$ with no children (i.e. all its children and outgoing edges have been deleted) clearly has $0$ accepting suffixes, or $|A_D(w)| = 0$, and so we delete $w$, along with $E(v,w)$. For the remaining children, we can compute $|A_D(v)| = \sum_{w \in L(D,\el + 1)}|E(v,w)||A_D(w)| = \sum_{w \in L(D,\el + 1)} (r_{vw} - \ell_{vw} + 1)|A_D(w)|$ as the number of accepting suffixes, and we compute 
$p_v(d)={|A_D(D(v,d))|}/{|A_D(v)|}$ for each remaining $d \in \{0,1,\dots,u-1\}$ which provides the distribution over accepting solutions below $v$. Note that the probability of going to vertex $w$ from $v$ is exactly 
${|E(v,w)||A_D(w)|}/{|A_D(v)|}$ as expected, since each of the $|E(v,w)|$ edge labels leading to $w$ has probability
$|A_D(w)|/|A_D(v)|$.

We can generate samples from this source by taking a random walk on $D$ according to $p_v$ and outputting the labels of the edges traversed.

Now we build the rounded ROBPs for each constraint using this source.
Recall that section \ref{sec:sinlge-robp} tells us that the ROBP $M^i$ exactly computing the indicator function for each $Z_i$ is an interval ROBP. Recall that each layer $\el$ of $M^i$ has at most $u^n$ states. We represent state $v$ by partial sum $\sum_{j = 1}^{\el}f_{ij}(x_j)$ where $x$ is the string of edge labels used to reach $v$ from start state $s$. 
We will use the following theorems to build an approximate ROBP $\M^i$ for each $Z_i$ individually.
\begin{theorem} \label{mainthm}
Given a $(W,n)$-interval ROBP $M^i$ for set $Z_i(f_i,b_i,u)$, $\delta > 0$, and a small-space distribution $D$ over $\{0,1,\dots,u-1\}^n$ of width at most $s = n^{O(k)}$, there exists an $(O(ns\log u/\delta),n)$-width interval ROBP $\M^i$ such that for all $z$, $M^i(z) \leq \M^i(z)$ and \[\Pr_{x \from D}[M^i(z) = 1] \leq \Pr_{x \from D}[\M^i(z) = 1] \leq (1 + \delta )\Pr_{x \from D}[M^i(z) = 1]\] Moreover, given an implicit description of $M^i$ and an explicit description of $D$, $\M^i$ can be constructed in deterministic time $O(n^{O(k)}(\log u)^2\log b/\delta^2)$
\end{theorem}
\begin{proof}
 We will construct a new small width ROBP $\M^i$ by iteratively \enquote{rounding} each layer of $M^i$, starting from the last one (layer $n$).   
As we did in section \ref{sec:sinlge-robp}, we start from the exact branching program $M^i$ and construct a sequence of programs $\M^i_n = M^i, \ldots, \M^i_0 = \M^i$, where $\M^i_j$ is obtained from $\M^i_{j+1}$ by rounding layer $(j+1)$. Again, we assume that $L(M,\el)$ has \emph{all} values in $[0,\sum_{j = 1}^{\el}f_j(u-1)]$ that have encoding size at most $L' = L + \log u + \log n$.

    We do the rounding in such a way that the acceptance probabilities are well approximated under each of the possible distributions on suffixes $D^w$. The program $\M^i$ will have polynomial width.
    
    We construct each new ROBP $\M^i_{j}$ as follows:\\
    Consider the ordered partial sums in $L(\M^i_{j + 1},j + 1)$, which is the set of vertices we wish to round in this iteration. Fix a vertex $w \in \W^{j + 1 }$
    . We wish to select a subset of $L(\M^i_{j + 1},j + 1)$ to remain in $\M^i_{j}$.
    We will do so by contracting vertices with \enquote{similar} acceptance probabilities. We create groups of \enquote{similar} vertices by finding breakpoints where acceptance probabilities change by a factor of $1/{(1 + \delta)}$. Vertices between breakpoints will be contracted to the neighboring breakpoint with higher acceptance probability. However, unlike before, we will calculate probability by sampling over accepting suffixes for each $w \in \W^{j + 1}$.
    
    Formally, we define a set $B^{j + 1}(w) = \{\beta^w_{1}, \dots,\beta^w_{N_w +1} \} \seq L(\M^i_{j + 1},j + 1)$ of breakpoints for $w$ as follows. We start with $\beta^w_{1} =  0$, the vertex with highest acceptance probability (and lowest partial sum), and given $\beta^w_{\el}$, define $\beta^w_{\el+1}$ by 
    \[\beta^w_{\el+1} = \max \big\{ v : v \prec \beta^w_{\el} \text{ and } 0 < P_{\M^i_{j},w}(v) <  P_{\M^i_{j},w}(\beta^w_{\el}) / (1 + \delta)\big\},\]
    which is the smallest partial sum whose acceptance probability decreases by a factor of at least $1/{(1 + \delta)}$ from the previous breakpoint. Now consider $\beta^w_{N_w}$, the last breakpoint assigned. Let $\beta^w_{N_w + 1} = \min \{v \in L(\M^i_{j + 1}, j + 1) : P_{\M^i_{j + 1},w}(v) = 0 \}$. Note that $\beta^w_1 \succ \beta^w_2 \succ \dots \succ \beta^w_{N_w} \succ \beta^w_{N_w + 1} $ as the probabilities decrease.
    We set $L(\M^i_{j}, j + 1 ) = \{\beta_1 , \dots , \beta_{N}\} = B^{j + 1}:=  \cup_{w \in \W^{j + 1}}B^{j + 1}(w)$ which is the union of breakpoints over all $w$. 
    Now we have intervals $I_1 = [\beta_1,\beta_2 -1/2^{2L'}], I_2 = [\beta_2,  \beta_3 -1/2^{2L'}],\dots, I_{N} = 
    [\beta_N , v_W]$.
    
The vertices in all other layers stay the same as in $\M^i_{j + 1}$, as do all the edges except those from layer $j$ to $j + 1$. We round these edges \emph{upward} as follows: for $v \in L(\M^i_{j}, j)$ and suffix $z \in \{0,1,\dots,u_{j + 1}\}$, we redirect the edge labeled $z$ to the breakpoint $\beta_m \in L(\M^i_{j}, j + 1)$ such that $\beta_{m+1} \prec v_z \preceq \beta_m$ where $v_z = v + f_{i(j + 1)}(z)$ is the true child of $v$. Again, we will represent edges $E(v,v')$ between vertices $v, v'$ only by the extremes $\ell_{vv'}, r_{vv'}$ such that the only suffixes $z \in \{0,1, \dots, u_{j}\}$ that satisfy $\M^i_{j + 1}(v,z) = v'$ are $\ell_{vv'} \leq z \leq r_{vv'}$. 

We will not explicitly build these ROBPs, we will just use them to assist in our proofs, which use induction over the ROBPs $\M^i_{n}, \M^i_{n-1},\dots, \M^i_0$, and thus prove the claim for $\M^i_0 = \M^i$ as needed. 
    \\
    We prove that each $\M^i_j$ remains an interval ROBP through the following claims:
\end{proof}
\begin{claim}
    For $v \in L(\M^i_j,\el)$ and $0 \leq \ell < r \leq u_{\el + 1}$, $\M^{i}_j(v,\ell) \succeq \M^i_j(v,r)$.
\end{claim}
\begin{proof}
    Note that this claim is only interesting if $\el > j$, otherwise all vertices $v, \M^{i}_j(v,\ell), \M^i_j(v,r)$ exist in original interval ROBP $M^i$ and thus the claim holds. Thus we can assume that layer $\el + 1$ is rounded and consists of breakpoints. Take $b_{x_\ell} = \M^i_j(v, \ell)$ and $b_{x_r} = \M^i_j(v, r)$

    Note that for any \enquote{true} child $v_z = \M^i_{j + 1}(v,z)$, we replace the child through edge $z$ with breakpoint $\beta_{x_z}$ such that $\beta_{x_z}\succeq v_z \succ \beta_{x_z+1}$. This is exactly the the maximum partial sum breakpoint with lower partial sum than $v_z$. Thus, we can follow identical logic to claim \ref{A.2.1} to prove the claim.
\end{proof}
\begin{claim}
    Let $v,v' \in L(\M^i_j,\el)$, $v \succeq v'$. For suffix $z$, $\M^i_j(v,z) \succeq \M^i_j(v',z)$.
\end{claim}
\begin{proof}
    Note that we need only prove this claim for suffixes of length one, as the property for any longer suffix will follow inductively.
    
    Also note as in the previous theorem, this claim is only interesting if $\el > j$, otherwise vertices $v, v',\M^i_j(v,z), \M^i_j(v',z)$ exist in original interval ROBP $M^i$. Thus we can assume that layer $\el + 1$ is rounded and consists of breakpoints.

    Again since edges are re-assigned to breakpoints in the same fashion as section \ref{sec:sinlge-robp}, we can follow identical logic to claim \ref{A.2.2} to prove the claim.
\end{proof}
We now analyze how the probability of acceptance (and thus the total number of solutions) in our new ROBP $\M^i$ compares to that of the original ROBP $M^i$. 
\begin{theorem} \label{subsetthm}
For $v \in L(\M^i_j,\el)$, we have $A_{\M^i_{j+ 1}}(v) \seq A_{\M^i_j}(v)$. 
\end{theorem}
\begin{proof}
    Note that if $\el > j$, both ROBPs make identical transitions from $v$, and $A_{\M^i_{j+ 1}}(v) = A_{\M^i_j}(v)$. 
    
    Assume $\el = j$. Let $\beta_z = \M^i_j(v,z)$ be the children of $v$ for each $z \in \{0,1,\dots,u_{\el + 1}\}$. Note that since $\M^i_j$ is obtained from $\M^i_{j + 1}$ by rounding layer $j + 1$, these vertices are breakpoints, and for the \enquote{true} children of $v$, there are vertices $v_z = \M^i_{j + 1}(v,z)$ in layer $j + 1$ of $\M^i_{j + 1}$ such that $v_z \preceq \beta_z$. Note that by our previous theorems, suffixes maintain ordering, and thus $A_{\M^i_{j+1}}(\M^i_{j+1}(v,z)) = A_{\M^i_{j+1}}(v_z) \seq A_{\M^i_{j}}(\beta_z) = A_{\M^i_{j+1}}(\M^i_{j}(v,z))$. Thus the set of accepting suffixes can only increase for any choice of $z$, and so the claim holds. 
    
    Otherwise, if $\el < j$ the claim follows because $\M^i_j$ and $\M^i_{j + 1}$ make identical transitions up to layer $j$, and thus the claim follows here as well. 
\end{proof}
\begin{theorem} \label{probthm}
For any $v \in L(\M^i_j,\ell)$ and $w \in \W^\ell$, we have that $P_{M^i,w}(v) \leq P_{\M^i_j,w}(v) \leq P_{M^i,w}(v)(1 + \delta)^{n - \ell}$.
\end{theorem}
\begin{proof}
Note that since $A_{\M^i_{j+1}}(v) \seq A_{\M^i_j}(v)$, it is clear that $P_{\M^i_{j + 1},w}(v) \leq P_{\M^i_j,w}(v)$ for all $w$, and thus $P_{M^i, w}(v) = P_{\M^i_n,w}(v) \leq P_{\M^i_{n-1},w}(v)\leq \dots \leq P_{\M^i_j,w}(v)$
    
    Now we will prove the right-hand side inequality by induction on the ROBPs $\M^i_{n}, \M^i_{n-1},\dots, \M^i_0$, by showing that $P_{\M^i_j,w}(v) \leq P_{\M^i_{j+1},w}(v)(1 + \delta)$ for every $v \in L(\M^i_j,\ell)$, $w \in \W^\ell$.
    
    Note that when proving this claim for $\M^i_j$, any layer $\ell > j$ is identical in $\M^i_{j+1}$, and so the claim is obvious. Also note that if we prove the claim for vertices in layer $\ell = j$, the claim will follow for layers in $\ell < j$, because these layers are also identical in $\M^i_{j+1}$, and thus make identical transitions and consequently we can express both $P_{\M^i_j,w}(v)$ and $P_{\M^i_{j+1},w}(v)$ as the same convex combination of acceptance probabilities of vertices in layer $j$. 
    
    Let the layer $\ell = j$, and consider vertex $v \in L(\M^i_j , j)$. 
    We find that probability $P_{\M^i_j,w}(v)$ is equal to the sum, over edges $z \in \{0,1, \dots, u_{j+1}\}$ out of $v$, of the probability we take that same edge label $z$ out of $w$, multiplied with the probability of accepting if we start from $\M^i_j(v, z)$ and make transitions in $\M^i_j$ according to a suffix sampled from $D^{w_z}$. Here, $w_z = D(w,z)$ is the child of $w$ in $D$ through edge labeled $z$. Explicitly,  
    \[P_{\M^i_j,w}(v) = \sum_{z = 0}^{u_{j + 1}} p_w(z) P_{\M^i_j,w_z}(\M^i_j (v,z))\]
     We now aim to bound $P_{\M^i_j,w_z}(\M^i_j(v, z))$ for the children of $v$ taken through each edge label $z$. Let $\beta_{1}, \beta_4$ be the adjacent breakpoints in $B^{i + 1}(w_z)$ such that $\beta_1 \prec v + f_{i(j+1)}(z) \preceq \beta_4$, and let $\beta_{2}, \beta_3$ be the adjacent breakpoints in $B^{i + 1}$ such that $\beta_2 \prec v + f_{i(j+1)}(z) \preceq \beta_3$, and note that $\M^i_j(v,z) = \beta_3$ by the way we connected layer $j$ to $j + 1$ in $\M^i_j$. Also note that since both layer $j$ and $j + 1$ of $\M^i_{j + 1}$ are identical to $\M^i$, we have that $v + f_{i(j+1)}(z) = \M^i_{j + 1}(v,z)$.  \\
     Since $B^{i + 1}(w_z) \seq B^{i + 1 }$ we have $\beta_1 \preceq \beta_2 \prec \M^i_{j + 1}(v,z) \preceq \beta_3\preceq \beta_4$.
     By the way we chose breakpoints and assign edges when rounding layer $j + 1$, we have that 
     \[P_{\M^i_{j + 1}, w_z}(\beta_4) \leq (1 + \delta)P_{\M^i_{j + 1}, w_z}(\M^i_{j + 1}(v,z))\]
     Since $\M^i_{j + 1}$ is an interval ROBP, we have that $P_{\M^i_{j + 1}, w_z}(\beta_3) \leq P_{\M^i_{j + 1}, w_z}(\beta_4)$ and thus 
     \[P_{\M^i_{j + 1}, w_z}(\beta_3)  \leq (1 + \delta)P_{\M^i_{j + 1}, w_z}(\M^i_{j + 1}(v,z))\]
     Since the breakpoints are a subset of the vertices in layer $j + 1$ of $\M^i_{j + 1}$, they exist in both $\M^i_j$ and $\M^i_{j + 1}$, and thus $P_{\M^i_{j + 1}, w_z}(\beta_3) = P_{\M^i_{j}, w_z}(\beta_3)$. Recall that $\M^i_j(v,z) = \beta_3$, and thus 
     \[P_{\M^i_j, w_z}(\M^i_j (v,z)) \leq (1 + \delta)P_{\M^i_{j + 1}, w_z}(\M^i_{j + 1}(v,z)) \]
     Thus, we have 
    \[P_{\M^i_j,w}(v) \leq \sum_{z = 0}^{u_{j+1}} p_w(z) (1 + \delta )P_{\M^i_{j+1},w_z}(\M^i_{j+1}(v,z)) = (1 + \delta) P_{\M^i_{j+1},w}(v)\]
     Thus we have proven the claim, and it follows that 
     \[P_{\M^i_j,w}(v) \leq P_{\M^i_{j + 1},w}(v)(1 + \delta) \leq \dots \leq P_{\M^i_{n},w}(v)(1 + \delta)^{n- j} = P_{\M^i,w}(v)(1 + \delta)^{n- \ell}\] for every $v \in L(\M^i_j,\ell)$, $w \in \W^{\ell}$.
\end{proof}
The above implies that for every $v\in L(\M^i,\ell) = L(\M^i_0,\ell)$, and every $w \in \W^\ell$ we have  \[P_{\M^i,w}(v) \leq P_{M^i,w}(v)(1 + \delta)^{n}\]
Note that we have Taylor expansion $(1 + \delta)^n = 1 + n\delta + \frac{n(n-1)}2\delta^2 + \dots$, thus if we take $\delta = \frac\eta {2n}$, which is small, the second term $n\delta$ is much larger than the remainder of the terms, and we have $(1 + \delta)^n \leq 1 + 2n\delta = 1 + 2n\frac\eta {2n} = 1 + \eta$, and thus \[P_{M^i,w}(v) \leq P_{\M^i,w}(v) \leq P_{M^i,w}(v)(1 + \eta)\]
for every $v$ and $w$ as we will use in the following theorems. 
\begin{theorem} \label{runtime}
We can construct ROBP $\M^i$ using $O({n^{O(k)} (\log u)^3  (L + \log n)} \log ({n^{O(k)} \log u}/\eta)/\eta^2 ) $ arithmetic operations on $O(L\log u + \log n)$-bit numbers
\end{theorem}
\begin{proof}
    Observe that for every vertex $v$ in layer $\ell$ and every $w \in \W^\ell$ we have that $P_{\M^i, w}(v) \leq 1$ and $P_{\M^i, w}(v) \geq u^{-n}$. Let $N_w+1$ be the number of breakpoints in $B^\ell(w)$. For $1 < j \leq N_w$, every breakpoint $\beta^w_{j}$ changes from $\beta^w_{j-1}$ by at least a factor of $\frac1{1 + \delta}$, and thus \[u^{-n} \leq P_{\M^i, w}(\beta^w_{N_w}) < P_{\M^i, w}(\beta^w_{1})/(1 + \delta)^{N_w} \leq (\frac1{1 + \delta})^{N_w}\]
    This implies that $u^n \geq (1 + \delta)^{N_w}$, and thus $\log u \cdot n \geq \log (1 + \delta) \cdot {N_w}$, or $N_w \leq {n \log u }/{\log(1 + \delta)}$. Taylor series approximation tells us that for $\delta < 1$, we have \[\log(1 + \delta) \geq \delta - \frac{\delta^2}2\quad \implies\quad N_w \leq \frac{n \log U}{\delta - \frac{\delta^2}2}\leq \frac{n \log U}{\delta - \frac{\delta}2}\leq \frac{n \log U}{\frac{\delta}2} \leq \frac{2n \log U}{{\delta}},\] 
    and thus $|B^{\ell}(w)| \leq 1 + {2n \log U}/{{\delta}} $, and $|B^\ell| = |\cup_{w \in \W^\ell}B^{\ell}(w)| \leq \sum_{w \in \W^\ell}|B^\ell(w)| \leq   s + {2ns \log u}/{{\delta}} =   O({n^2s \log u}/{{\eta}} )$, where $s$ is the width of small space source $D$ and $\eta = \del/2n$.
    
     Now note that for layer $\ell = n$, we create two vertices, $\beta_1 = 0$ with $P_{\M, w}(0) = 1$ for all $w \in \W^n$ and $\beta_2 = b_i + 1/2^{2L'}$ with $P_{\M, w}(b_i + 1/2^{2L'}) = 0$ for all $w \in \W^n$ in constant time. 
     
     We will now discuss the runtime of building layer $\ell$ from layer $\ell + 1$. We maintain that the vertices of the layer below, $L(\M^i, \ell+1) = B^{\ell+1}$, are known and stored in a binary tree along with the values $P_{\M^i, w}(\beta)$ for every $\beta \in B^{\ell+1}$ and $w \in \W^{\ell+1}$. We also maintain that each breakpoint has encoding length at most $L'$.
     
     We now prove that the statement holds for any $w \in \W^\ell$.
     For fixed $w \in \W^\ell$, we can compute the first breakpoint vertex of layer $\ell$ of $\M^i$ as $\beta^w_{1} = 0$. Now assume inductively that we have breakpoints $\beta^w_{1}, \dots, \beta^w_{j}$, and we now aim to compute $\beta^w_{j+1}$. 
     
     Recall that $\beta^w_{j+1}$ is the maximal $v \prec \beta^w_{j}$ satisfying $P_{\M^i , w }(v) < P_{\M^i , w }(\beta^w_{j})/(1 + \delta)$. 
    Thus, to find $\beta^w_{j+1}$
    we can do binary search on values in $[0,b]$, computing their associated $P_{\M_{j+1}}(v)$ each time. This range is acceptable because
we require $0 < P_{\M_{j + 1}}(v)$, consequently we need only search through vertices in the range $[0,b]$ since any value larger than $b$ can never lead to an accepting suffix due to nonnegativity of each $f_j$. We will only consider breakpoints with bit size $L'$ and this requires
$ O(\log b_i + L \log u + \log n)$ iterations.

    It remains to analyze the time to calculate $P_{\M^i , w }(v)$ and check if it exceeds $P_{\M^i , w }(v_{w(j)})/(1 + \delta)$ for each $v$ encountered during the binary search. 
     
    Here we have already calculated and stored $P_{\M^i , w }(v_{w(j)})$. We can calculate \[P_{\M^i,w}(v) = \sum_{z = 0}^{u_{j + 1}} p_w(z) P_{\M^i,w_z}(\M^i (v,z))\] where $w_z = D(w, z) \in \W^{\ell + 1}$. 

    It is important to note that in order to calculate probabilities, we do not ever need to take this sum above over all $u$ values. This follows from the fact that both $\M$ and $D$ will have edge intervals, the number of which are bounded by their respective widths. Thus, taking the union over all endpoints of these intervals, we can create new intervals such that in both $\M$ and $D$, all edge labels act the same within the interval. For each such interval $I$, choosing any $z\in I$, its contribution to the sum is $|I|p_w(z)P_{\M^i,w_z}(\M^i(v,z))$. We take the sum over these new intervals instead of over all $u$ edge labels.
    Thus, given $\M^i$ of width $O(n^2s\log u/\eta)$ and $D$ of width $s$, we need only sum over $O(s + n^2s\log u/\eta)$ such intervals.
    
    Also note that since we store these known probabilities in a binary tree, we need a factor of $\log({4n^2s \log u}/\eta)$ to access the information. Thus we can calculate this probability in time $O({n^2s \log u} \log ({n^2s \log u}/\eta)/\eta)$.
     
     We repeat this process throughout the binary search, iterating through the ordered partial sums $O(\log b+ L \log u + \log n)$ times. Thus, to find the next unknown breakpoint, we have runtime $O({n^2s \log u (\log b+ L \log u + \log n)} \log ({n^2s \log u}/\eta)/\eta)$. As shown above, we have at most $O(n^2 \log u / \eta)$ breakpoints for each $w \in \W^\ell$, and there are at most $s = n^{O(k)}$ many such $w$, and thus we can find all vertices in this layer in time $O({n^{O(k)} (\log u)^2  (\log b+ L \log u + \log n)} \log ({n^{O(k)} \log u}/\eta)/\eta^2 ) $. 
     Thus we can build all layers in time $O({n^{O(k)} (\log u)^2  (\log b+ L \log u + \log n)} \log ({n^{O(k)} \log u}/\eta)/\eta^2 ) $. Recall that we only compute breakpoints to $L'$-bit accuracy, and consequently our calculations are all on $L\log u + \log n$-bit numbers. By the same logic as the previous section and claim \ref{encoding-gap}, this bit accuracy allows us to access all information needed.
\end{proof}
\begin{theorem}\label{buildM}
    Given a collection of the $(W,n)$-ROBPs $\M^i$ for $i \in [k]$ and $(s,n)$-small space source $D$ as described above, we can create a $(W^k, n)$- ROBP $\M$ computing the intersection of these, i.e. for any $x \in \{0,1,\dots, u-1\}^n $, $\M(x) = 1$ if and only if $\M^i(x) = 1$ for all $i \in [k]$. 
\end{theorem}
\begin{proof}
    First, to create the desired ROBP $\M$, we need $n+1$ layers, where layer $0$ has a single vertex $(0,0,\dots,0)$ which is a $k$-tuple of partial sum $0$, as $0$ is the partial sum of the start vertex in each $\M^i$. Now from layer $\ell-1$ we will build layer $\ell$ as follows:
    
    For each vertex $v = (v_1, v_2, \dots, v_k) \in L(\M,\ell-1)$ represented by a $k$-tuple of partial sums, and for each $z \in \{0, 1 , \dots, u-1\}$, we construct child vertex through edge labeled $z$, which has partial sum tuple $(\M^1(v_1,z),(\M^2(v_2,z), \dots, (\M^k(v_k,z) )$. We maintain that the $i^{th}$ member of this tuple is associated with a vertex in $\M^i$, and that the prefix taken in $\M$ to reach vertex $v = (v_1, \dots, v_k)$ is exactly the prefix taken in $\M^i$ to reach $v_i$.
    
    Note that since there are at most $W$ vertices in each layer of $\M^i$, we can make at most $W^k$ partial sum tuple vertices in each layer of $\M$. Thus $\M$ has width $W^k$, and can be generated in time $O(W^k)$. Also note that, as in each individual constraint ROBP, we describe edges out of $v$ through sets of suffixes $E(v,w)$ which lead to $w$. This bounds our edges by $W^k$ as well.
    
    Now to generate the probabilities in this new ROBP, we wish to maintain that for any vertex $v$ in any layer $L(\M, \ell)$ and any $w \in \W^\ell$, the probability of accepting if we start from $v = (v_1, v_2, \dots, v_k)$ and make transitions in $\M$ according to a suffix sampled from $D^w$ is equal to the probability of the following happening for all $i \in [k]$: we accept starting from $v_i$ and making transitions in $\M^i$ according to a suffix sampled from $D^w$.
    
    We can compute these starting from layer $n$.
    
    At layer $n$ we know that for any $w \in \W^n$, vertex $v = (v_1, \dots,v_k)$ only accepts if all $v_1, \dots v_k$ accept, i.e. they have probability $1$ of accepting in their respective $\M^i$s. If there exists any $i$ such that $v_i$ does not have probability $1$ of accepting in $\M^i$ (i.e. is not the vertex $0$), then we reject and have probability $0$.

    For any layer above, we assume we have the probabilities $P_{\M,w}$ computed for every vertex in the layers below. For vertex $v \in L(\M, \ell)$, and for each $w \in \W^\ell$ we can compute
    \[
    P_{\M,w}(v) = \sum_{z = 0}^{u-1} p_w(z) P_{\M,w_z}(\M(v,z)).
    \]
    As above, we need not sum over all $u$ edge labels. Taking the union of the endpoints of the edge intervals of $\M$ and $D$, we obtain refined intervals on which both $\M(v,z)$ and $D(w,z)$ are constant. For each such interval $I$, choosing any $z\in I$, its contribution to the sum is
    \[
    |I|p_w(z)P_{\M,w_z}(\M(v,z)).
    \]
    There are at most $O(W^k+s)$ such refined intervals. In the application below, $s\leq W^k$, so this is $O(W^k)$. Thus we calculate the probabilities for every $w$ and every $v$ in a layer in $O(W^{2k}s)$ time and thus the entire tree (all $n$ layers) in time $O(nW^{2k}s)$.
\end{proof}

Finally, we can prove that this ROBP rounding algorithm provides the guarantee of theorem \ref{multi-ROBP}.
\begin{proof}[Proof of theorem \ref{multi-ROBP}]
 First we call on theorem \ref{Dyer} to construct altered sets $S_i$. 
 Since our sets have width at most $\sum_{j = 1}^n h_{ij}(u_j) + 2n^2 \leq 2n^3$, it follows that the ROBP $M_{S}$ computing the indicator function over $S$ has width $O(n^{3k})$. This is because we can represent each state in layer $\el$ of $M_{S}$ as a $k$-tuple of partial sums $\{\sum_{j = 1}^{\el} h_{1j}(x_j) , \dots, \sum_{j = 1}^{\el} h_{kj}(x_j) \}$ for which there are at most $O(n^{3k})$ possible values.
 We use the uniform distribution over solutions in $S = \cap_{i \in [k]}S_i$ as our small space source $D$. 
Note that claim \ref{2.4} verifies this is a small space source of width $n^{O(k)}$. 

For $i \in [k]$, let $M^i$ be a $(W,n)$-ROBP exactly computing the indicator function for $Z_i $. Let $\eta = O(\eps / k(n+1)^k)$. Now, for every $i \in [k]$, by theorem \ref{mainthm}, in time $O(n^{O(k)}(\log u)^2(L + \log n)/\eta^2)$ we can explicitly construct a $(n^{O(k)}\log u/\eta, n)$-width ROBP $\M^i$ such that
\[\Pr_{x \leftarrow D}[\M^i(x) \neq M^i(x)] \leq \eta\]
Let $\M$ be the $(n^{O(k^2)}(\log u)^k/\eta^k, n)$-width ROBP computing the intersection of all $\M^i$ for $i \in [k]$ as described in Theorem \ref{buildM}. We have $\M(x) = \wedge_i \M^i(x),$ and thus by union bound \[ \Pr_{x \leftarrow D} [\M(x) \neq \wedge_i M^i(x)] \leq k\eta \]
On the other hand, by Theorem \ref{Dyer},
\[\Pr_{x \leftarrow D}[\bigwedge_iM^i(x) = 1] \geq \frac1{2(2n)^k} \]
Therefore, from the above two equations and setting $\eta = \eps/{4k(2n)^k}$, we get that
\[\Pr_{x \leftarrow D}[\bigwedge_iM^i(x) = 1] \leq \Pr_{x \leftarrow D}[\M(x) = 1] \leq (1 + \eps)\Pr_{x \leftarrow D}[\bigwedge_iM^i(x) = 1] \]
Thus, \[\Pr_{x \in_u \{0,1,\dots,u-1\}^n}[x \in S] \cdot \Pr_{x \from D}[\M(x) = 1] = \frac{|S|}{u^n}\cdot \Pr_{x \from D}[\M(x) = 1] \] is an $\eps$-relative error approximation to the fraction of solutions to all constraints \\ $\Pr_{x \in_u \{0,\dots,u-1\}^n}[\wedge_i M^i(x) = 1] = \Pr _{x \in_u \{0,\dots,u-1\}^n}[x \in S]\cdot \Pr_{x \from D}[\wedge_i M^i(x) = 1]$.

 The theorem now follows since we can compute ${|S|}/u^n$ in time $n^{O(k)}$ and using theorem \ref{buildM} we can compute $\Pr_{x \from D}[\M(x) = 1]$ using
 $O(n^{O(k^2)}(\log u/\eps)^{O(k)}(L + \log n))$ arithmetic operations on $O(L\log u + \log n)$-bit numbers. Multiplying this approximation by $u^n$ gives the desired $\eps$-relative error approximation to $|Z|$.
\end{proof} 

\pagebreak
\end{document}